\documentclass[12pt]{article} %
\usepackage{amsmath,amssymb,amsthm}
\numberwithin{equation}{section} %
\theoremstyle{plain}
    \newtheorem{thm}{\hspace{\parindent}{\sc Theorem}}[section] %
    \newtheorem{pro}[thm]{\hspace{\parindent}Proposition}
    
    \newtheorem{lem}[thm]{\hspace{\parindent}Lemma}
\theoremstyle{remark} %
    \newtheorem{rem}{\hspace{\parindent}Remark}[section] %
\newcommand{\alk}{a_{\mathrm{l}k}}
\newcommand{\alki}{a_{\mathrm{l}k}^{(\mathrm{i})}}
\newcommand{\alkone}{a_{\mathrm{l}k}^{(\mathrm{1})}}
\newcommand{\alktwo}{a_{\mathrm{l}k}^{(\mathrm{2})}}
\newcommand{\aonek}{a_{\mathrm{1}k}}

\newcommand{\atwok}{a_{\mathrm{2}k}}

\newcommand{\bDelta}{{\bold \Delta}}
\newcommand{\Cspace}{C^{\infty}}
\newcommand{\etavec}{\overrightarrow{\eta}}
\newcommand{\qvecdelta}{\overrightarrow{q}_{\Delta}}

\newcommand{\phiki}{\phi^{(\mathrm{i})}_k}
\newcommand{\rhoki}{\rho^{(\mathrm{i})}_k}

\newcommand{\Cts}{{\cal C}(t,s)}
\newcommand{\Cdelta}{{\cal C}_{\Delta}}

\def\hbar{{\mbox{\raisebox{-2pt}{$\mathchar'26$}}\mkern-9.5muh}}

\newcommand{\eonevec}{\overrightarrow{e_1}}
\newcommand{\etwovec}{\overrightarrow{e_2}}
\newcommand{\elvec}{\overrightarrow{e_j}}
\newcommand{\qts}{\overrightarrow{q}^{t,s}_{x,y}}
\newcommand{\qvec}{\overrightarrow{q}}

\newcommand{\Sspace}{{\cal S}}

\newcommand{\xvec}{\overrightarrow{x}}

\newcommand{\yvec}{\overrightarrow{y}}

\newcommand{\xivec}{\overrightarrow{\xi}}

\begin{document}
\title{Mathematical Remarks on the Feynman Path Integral
for  Nonrelativistic Quantum Electrodynamics}
\author{Wataru Ichinose\thanks{Research partially supported by Grant-in-Aid for Scientific Research No.16540145 and
No.19540175,  Ministry of Education, Culture, Sports, Science and 
Technology, Japanese
Government.}}
\date{}
\maketitle
\begin{quote}
{\small Department of Mathematical Science, Shinshu University,
Matsumoto 390-8621, Japan. \\
   E-mail: ichinose@math.shinshu-u.ac.jp}%
\end{quote}
\vspace{0.5cm}
{\small  {\bf Abstract:}   The Feynman path 
integral for nonrelativistic quantum electrodynamics is studied  mathematically of a 
standard model in physics, where the electromagnetic potential is assumed to be periodic 
with respect to a large box and quantized thorough its Fourier coefficients.
In physics, the Feynman path integral for nonrelativistic quantum 
electrodynamics is defined very formally. For example, as is often seen, even
independent variables are not so clear.
First, the Feynman path integral is defined rigorously under the constraints 
familiar 
in physics.  Secondly, the Feynman path integral  is also defined rigorously without the 
constraints, which is stated in Feynman and Hibbs' 
book without any
comments.  So, our definition may be completely new.  Thirdly,
 the vacuum and the state of photons of momentums and  polarization 
states are expressed by means of  concrete functions  of variables consisting of the Fourier coefficients of the electromagnetic potential. Our results above 
have many 
applications as is seen in Feynman and Hibbs' book, though the applications are not rigorous so far.  It is also proved rigorously by means of the distribution theory that the Coulomb potentials between 
charged particles naturally appear in the Feynman path integral above. As is 
well known, this shows that photons give the Coulomb forth.
 }
\section{Introduction}
A number of mathematical results on the Feynman path integral for
  quantum mechanics have been obtained.  On the
other hand, the author doesn't know any  mathematical results on the
Feynman
path integral for  quantum
electrodynamics (cf. \cite{Johnson-Lapidus}), written as QED from now on. \par
   A functional integral representation for a nonrelativistic QED model in 
   \cite{Pauli-Fierz} with
imaginary time
     was obtained by Hiroshima \cite{Hiroshima}  on the Fock 
spaces  in terms of the 
probabilistic method.
It has been well known that  the only translation invariant measure on a 
separable infinite
dimensional Banach space is the identically zero measure (cf. Theorem
4 in \S 5 of Chapter 4 in \cite{Gelfand-Vilenkin}). The measure
defining the Feynman path integral is expected to be translation invariant 
(cf. (7-29) in
\cite{Feynman-Hibbs}). So, there exist no measures defining the Feynman 
path integral.
\par
In the present paper the Feynman path integral for  nonrelativistic QED is 
studied rigorously of a 
standard model in physics (cf.
\cite{Dirac,Fermi,Feynman 50,Feynman-Hibbs,Pauli-Fierz,Sakurai}), where
the electromagnetic potential is assumed to be periodic with respect to a large box in $R^3$ and  quantized thorough
its Fourier coefficients.  In physics, the Feynman path integral for 
nonrelativistic QED is defined very formally.  For example, as is often seen 
(cf. \cite{Feynman-Hibbs}),
even independent variables are not so clear. 
     Our aim in the present paper is to  give  the mathematical definition of
the
Feynman path integral for  nonrelativistic QED of a standard model in physics.
We note that in the present paper, regrettably, 
   the Fourier coefficients with large wave numbers
  need to be
arbitrarily  cut off and we don't take the limit of  a box to $R^3$.
We also note that our nonrelativistic QED model is completely
different from  nonrelativistic QED models  on the Fock spaces (cf.
\cite{Gustafson-Sigal,Hiroshima,Spohn}). 
\par
    First,  the mathematical definition of the Feynman path integral
for
nonrelativistic QED is given under the constraints.  These
constraints  are  well known (cf.
(9-17) in \cite{Feynman-Hibbs}, (A-7) in
   \cite{Sakurai}, (13.10) in  \cite{Spohn} and (7.38) in
\cite{Swanson}). 
\par
Secondly, without the constraints  we give the
mathematical definition of the Feynman path
integral for
nonrelativistic QED, which has been  given by (9-98) in
\cite{Feynman-Hibbs} without any comments. Our method of giving the Feynman 
path
integral for
nonrelativistic QED without the constraints
  is like one
used in \cite{Ichinose 2000} for giving the phase space Feynman path integral.
   The author emphasize that any  definitions of
(9-98) in
\cite{Feynman-Hibbs} have not been given.  So our result may be completely new.
We  note
that our Feynman path integral without the constraints
   is  proved to be equal with the
Feynman path integral under the constraints  before taking the limit of the
discretization parameter.
\par
Thirdly, the vacuum and the states of photons of momentums and 
polarization states are expressed
by means of  concrete functions  of variables consisting of the Fourier 
coefficients of the
electromagnetic potential. In \cite{Feynman-Hibbs} only the vacuum and the
state of a photon with a momentum and a polarization state are expressed by 
means of the
concrete functions, which our functions are equal to. Generally, in physics 
the vacuum and the
states of photons with momentums and polarization states are not considered 
concretely but
considered abstractly (cf.
\cite{Sakurai,Swanson}). To write down the state of photons concretely, we 
introduce creation
operators and annihilation operators, which can be written concretely as 
partial differential
operators of the first order.
\par
   The results stated above have many applications as is seen in the chapter 9 
   of \cite{Feynman-Hibbs}, though the applications are not rigorous so far.
\par
Fourthly, we show by means of the distribution theory that 
the Coulomb potentials between charged particles 
appear when the periods
of the Fourier series tend to infinity and the cut-off of the Fourier 
coefficients is gotten
out.  This result, which shows that photons give the Coulomb forth, is well known in physics (cf. 
\cite{Fermi,Feynman-Hibbs}).  In the present
paper we give the rigorous proof.
\par
    The proof of giving a mathematical definition of the Feynman path 
integral for
nonrelativistic QED under or without the constraints is obtained by means of 
a somewhat
delicate study on  oscillatory integral operators, the abstract 
Accoli-Arzel\`{a} theorem on
the weighted Sobolev spaces and the uniqueness to the initial problem for 
the Schr\"odinger type
equations as in
\cite{Ichinose 1997,Ichinose 1999,Ichinose 2000,Ichinose 2003}.
\par
The proof of expressing the vacuum and the states of photons with momentums 
and polarization
states by means of  concrete functions is as follows.  We take 
$\mathrm{e}^{ik\cdot x}\
(k,x \in R^3)$ as the Fourier functions.  Then, annihilation operators are 
defined for the real
parts and the imaginary parts of the Fourier coefficients, respectively. 
Combining the annihilation 
operators for the real
parts and the imaginary parts, we can define 
the annihilation
operators of photons. The creation operators are defined as the adjoint 
operators of the
annihilation operators.
\par
The proof of the appearance of the Coulomb potentials between charged 
particles is given by
proving the convergence theorem for the Riemann sum of a unbounded function 
as the
discretization parameter tends to zero, which will be stated in Proposition 
4.3 in the present
paper.
\par
Our plan in the present paper is as follows. \S2 is devoted to
preliminaries.  In \S3  the main results on the
Feynman path integral for  nonrelativistic QED are stated. In \S4 the 
appearance of the
Coulomb potentials between charged particles is proved rigorously.
   In \S5 the vacuum and the states of photons with momentums and 
polarization state are
given concretely. \S6 -
\S9 are devoted to  proofs of the main results stated in \S3.
\section{Preliminaries}
      We consider $n$ charged nonrelativistic particles
      $x^{(j)} \in R^3\ (j = 1,2,\dotsc,n)$ with
mass $ m_j > 0$ and charge $e_j \in R.$ Let $T > 0$ be an arbitrary 
constant, $t \in [0,T], x
\in R^3,
\phi (t,x)
\in R$ a
scalar potential and $A(t,x) \in R^3$ a vector potential.  We
set
\begin{align*}
& \xvec := (x^{(1)},\dotsc,x^{(n)}) \in R^{3n}, \\
& \dot{\xvec} := (\dot{x}^{(1)},\dotsc,\dot{x}^{(n)}) \in R^{3n}.
\end{align*}
Then the Lagrangian function for    particles and the
electromagnetic field with the charge density
\begin{equation} \label{2.1}
\rho(t,x) = \sum_{j=1}^n  e_j   \delta \left(x - x^{(j)}(t)\right)
\end{equation}
and the current density
\begin{equation} \label{2.2}
j (t,x) = \sum_{j=1}^n  e_j   \dot{x}^{(j)}(t) \delta \left(x -
x^{(j)}(t)\right) \in R^3, \quad \dot{x}^{(j)}(t) = \frac{dx^{(j)}}{dt}(t)
\end{equation}
is given by
\begin{align} \label{2.3}
&{\cal L}\left(t,\xvec,\dot{\xvec},A,
\dot{A},
\frac{\partial A}{\partial x},\phi, \frac{\partial \phi}{\partial x}
\right) \notag \\
&= \sum_{j=1}^n
  \frac{m_j }{2}|\dot{x}^{(j)}|^2 - \int\rho (t,x)\phi (t,x)dx +
   \frac{1}{c} \int j(t,x) \cdot A(t,x)dx  \notag \\
& \quad + \frac{1}{8\pi}\int_{R^3}\left( |E(t,x)|^2 - |B(t,x)|^2 \right)dx +
C \notag \\
&= \sum_{j=1}^n
\left( \frac{m_j }{2}|\dot{x}^{(j)}|^2 - e_j\phi (t,x^{(j)})+
   \frac{1}{c} e_j \dot{x}^{(j)} \cdot A(t,x^{(j)}) \right)  \notag \\
& \quad + \frac{1}{8\pi}\int_{R^3}\left( |E(t,x)|^2 - |B(t,x)|^2 \right)dx +
C
\end{align}
(cf. \cite{Feynman-Hibbs}, \cite{Spohn}),
where
\begin{equation} \label{2.4}
E = -\frac{1}{c}\frac{\partial A}{\partial t} - \frac{\partial
\phi}{\partial x}, \quad
B = \nabla \times A,
\end{equation}
  $\partial \phi/\partial x = (\partial \phi/\partial x_1,\partial \phi/\partial
x_2,\partial \phi/\partial x_3)$ and $C$ is an indefinite constant.  It 
seems that an
indefinite constant in \eqref{2.3} has not been used by anyone before (cf.
\cite{Feynman-Hibbs,Sakurai,Spohn}).
\par
    As in \cite{Fermi,Feynman 50,Pauli-Fierz,Sakurai} we consider
a sufficient large box
\[V = \left[ -\frac{L_1}{2},\frac{L_1}{2}\right] \times \left[
-\frac{L_2}{2},\frac{L_2}{2}\right] \times \left[
-\frac{L_3}{2},\frac{L_3}{2}\right] \subset R^3.
\]
As variables we consider all  periodic potentials  $\phi (t,x)$ and
$A(t,x)$ in $x \in R^3$ with periods $L_1, L_2$ and $L_3$ satisfying
\begin{equation} \label{2.5}
\nabla \cdot A (t,x) = 0\  \ \text{in}\ [0,T]\times R^3\ \ \text{(the
Coulomb gauge)}
\end{equation}
and also
\begin{equation} \label{2.6}
\int_V \phi(t,x)dx = 0, \quad \int_V A(t,x)dx = 0.
\end{equation}
Let $|V| = L_1L_2L_3$. We set
\begin{equation} \label{2.7}
k :=
\left(\frac{2\pi}{L_1}s_1,\frac{2\pi}{L_2}s_2,
\frac{2\pi}{L_3} s_3 \right)
   \quad (s_1,s_2,s_3 = 0,\pm 1,\pm 2,\dotsc)
   \end{equation}
   and take
   $\elvec (k) \in R^3\  (j = 1,2)$ such that $\left(\overrightarrow{e_1}(k),
   \overrightarrow{e_2}(k), k/|k|\right)$ for all $k \not= 0$ forms a set of
mutually orthogonal unit vectors and
\begin{equation} \label{2.8}
  \overrightarrow{e_j}(-k) = - \overrightarrow{e_j}(k) \quad  (j = 1,2).
   \end{equation}
We note that we can easily determine measurable functions 
$\overrightarrow{e_j}(k)\in R^3\  (k
\in R^3, j =1,2)$ by the Gram and Schmidt method such that 
$\left(\overrightarrow{e_1}(k),
   \overrightarrow{e_2}(k), k/|k|\right)$ for all $k \not= 0$ forms a set of
mutually orthogonal unit vectors and satisfies \eqref{2.8} (cf. p. 448 in 
\cite{Arai}).
   Noting \eqref{2.5} and
\eqref{2.6}, we can expand $\phi(t,x)$ and $A(t,x)$ formally   into the 
Fourier series
\begin{align} \label{2.9}
& A(x,\{a_{\mathrm{l}k}(t)\})  = \frac{\sqrt{4\pi}}{|V|}c \sum_{k\not= 0}
\Bigl\{ \aonek (t) \mathrm{e}^{ik\cdot x} \eonevec (k) +
\atwok (t) \mathrm{e}^{ik\cdot x}\etwovec (k) \Bigr\},\\
   & \phi (x,\{\phi_k (t)\}) = \frac{1}{|V|} \sum_{k\not= 0}\phi_k(t)
   \mathrm{e}^{ik\cdot x}.
\end{align}
\begin{rem}
In physics (cf. \cite{Feynman-Hibbs,Sakurai}) the condition \eqref{2.6} is not 
assumed clearly.
\end{rem}
     We write
\begin{align} \label{2.11}
&\alk =:\frac{\alkone - i\alktwo}{\sqrt{2}}\quad  (\mathrm{l} = 1,2),
\\
& \phi_k  =:\phi_k^{(1)} - i\phi_k^{(2)},
\end{align}
where $\alki \in R$ and $ \phi_k^{(i)}  \in R$, and also the complex 
conjugate of $\alk$ as
$\alk^*$. Since $A$ and $\phi$ are real valued, the relations
\begin{equation} \label{2.13}
a_{\mathrm{l}-k}^{(\mathrm{1})} = - \alkone,\quad
a_{\mathrm{l}-k}^{(\mathrm{2})} =  \alktwo,\quad
\phi_{-k}^{(\mathrm{1})} = \phi_{k}^{(\mathrm{1})},\quad
\phi_{-k}^{(\mathrm{2})} = - \phi_{k}^{(\mathrm{2})}
\end{equation}
hold from \eqref{2.8}. So, from \eqref{2.9} and (2.10) we have
\begin{align} \label{2.14}
   &A(x,\{a_{\mathrm{l}k}\})  = \frac{\sqrt{4\pi}}{|V|}c \sum_{k\not= 0} 
\sum_{\mathrm{l}=1}^2
\frac{1}{\sqrt{2}}
\big( \alkone \cos k\cdot x + \alktwo \sin k\cdot x
\bigr)\overrightarrow{e_{\mathrm{l}}} (k), \\
& \phi(x,\{\phi_k\}) = \frac{1}{|V|} \sum_{k\not= 0}
\big( \phi^{(\mathrm{1})}_k \cos k\cdot x + \phi^{(\mathrm{2})}_k \sin
k\cdot x
\bigr).
\end{align}
We also write
\begin{align} \label{2.16}
&\rho^{(1)}_k(\xvec) := \sum_{j=1}^n e_j \cos k\cdot x^{(j)}, \\
&\rho^{(2)}_k(\xvec) := \sum_{j=1}^n e_j \sin k\cdot x^{(j)}.
\end{align}
\par
   Determining an indefinite constant $C$ in the Lagrangian function 
\eqref{2.3} formally by
\begin{equation} \label{2.18}
  \frac{2\pi}{|V|} \sum_{j=1}^n e_j^2 \sum_{k\not= 0}\frac{1}{|k|^2} +
\frac{1}{2}\sum_{k\not= 0}\frac{\hbar c |k|}{2},
\end{equation}
we can
write
${\cal L}$  from \eqref{2.3} by means of \eqref{2.4}, \eqref{2.9}, (2.10) 
and (2.15) as
\begin{align} \label{2.19}
&{\cal L}(\xvec,\dot{\xvec}, \{\alk\}, \{\dot{a}_{\mathrm{l}k}\},\{\phi_k\})
=
\sum_{j=1}^n
\frac{m_j }{2}|\dot{x}^{(j)}|^2   \notag \\
& \quad + \frac{1}{8\pi |V|}\sum_{k\not=0}
\Biggl\{\sum_{\mathrm{i}=1}^2\left(|k|^2 
\left(\phi^{(\mathrm{i})}_k\right)^2  -
8\pi
\rho^{(\mathrm{i})}_k(\xvec)\ \phiki \right)    \notag \\
& \quad + 16\pi^2
\frac{\sum_{j=1}^n e_j^2}{|k|^2}
\Biggr\} + \frac{1}{c}\sum_{j=1}^n e_j\ \dot{x}^{(j)}\cdot
A(x^{(j)},\{\alk\})  \notag \\
&  \quad + \frac{1}{ 2}\sum_{k\not=0,\mathrm{i,l}}
\left(\frac{\left(\dot{a}_{\mathrm{l}k}^{(\mathrm{i})}\right)^2}{2|V|} -
\frac{(c|k|)^2\left(\alki\right)^2}{2|V|}  + \frac{\hbar c|k|}{2}\right).
\end{align}
\begin{rem}
If we don't assume (2.6), we must add $(-1/|V|) (\sum_{j=1}^n
e_j)\phi^{(1)}_0$ and
$\sum_{\mathrm{i,l}=1,2}\left(\dot{a}_{\mathrm{l}0}^{(\mathrm{i})}\right)^2/ 
(4|V|)$ to
\eqref{2.19}.
\end{rem}

  The
reason why we determined an indefinite constant in
\eqref{2.3} by \eqref{2.18} will be explained in Remark 5.1. Taking account 
of the constraints
\begin{equation} \label{2.19'}
  |k|^2\phiki = 4\pi\rhoki(\xvec)\quad  (\mathrm{i} = 1,2, \ k \not= 0),
\end{equation}
roughly $\nabla\cdot E = 4\pi\rho$ from \eqref{2.1} and \eqref{2.4}
as in (9-17) in \cite{Feynman-Hibbs} and (7.38) in \cite{Swanson}, then 
from \eqref{2.16} and
(2.17) we have
\begin{align} \label{2.21}
&\sum_{\mathrm{i}=1}^2\left(|k|^2 \left(\phi^{(\mathrm{i})}_k\right)^2  -
8\pi
\rho^{(\mathrm{i})}_k(\xvec)\ \phiki \right)
  + 16\pi^2
\frac{\sum_{j=1}^n e_j^2}{|k|^2}  \notag \\
& = - \frac{16\pi^2}{|k|^2}\left\{\left(\rho_k^{(1)}\right)^2 + 
\left(\rho_k^{(2)}\right)^2 -
\sum_{j=1}^n e_j^2
\right\}
\notag \\
& = - \frac{16\pi^2}{|k|^2} \sum_{j,l =1, j\not= l}^n e_je_l 
\mathrm{e}^{ik\cdot
x^{(j)}}\mathrm{e}^{-ik\cdot x^{(l)}}
\notag \\
&  = - \frac{16\pi^2}{|k|^2} \sum_{j,l =1, j\not= l}^n
e_j e_l\cos k\cdot (x^{(j)} - x^{(l)}).
\end{align}
So we get
\begin{align} \label{2.22}
&{\cal L}_c(\xvec,\dot{\xvec}, \{\alk\}, \{\dot{a}_{\mathrm{l}k}\}) =
\sum_{j = 1}^n
\frac{m_j }{2}|\dot{x}^{(j)}|^2   \notag \\
&  \quad - \frac{2\pi}{|V|}\sum_{k \not= 0}\sum_{j,l =1,j\not=l}^n
\frac{e_j e_l\cos k\cdot (x^{(j)} - x^{(l)})}{|k|^2 } \notag \\
& \quad + \frac{1}{ c}\sum_{j=1}^n  e_j\/\dot{x}^{(j)}\cdot
A(x^{(j)},\{\alk\})  \notag \\
&  \quad + \frac{1}{2}\sum_{k \not= 0,\mathrm{i,l}}
\left(\frac{\left(\dot{a}_{\mathrm{l}k}^{(\mathrm{i})}\right)^2}{2|V|} -
\frac{(c|k|)^2\left(\alki\right)^2}{2|V|}  + \frac{\hbar c|k|}{2}\right).
\end{align}
\par
For a multi-index
$\alpha = (\alpha_1,\dotsc,\alpha_m)$ and $z = (z_1,\dotsc,z_m) \in R^m$ we 
write $|\alpha|
=
\sum_{j=1}^m
\alpha_j$, $z^{\alpha} =  z_1^{\alpha_1}
\cdots  z_m^{\alpha_m}, \partial_z^{\alpha} = (\partial /\partial
z_1)^{\alpha_1}
\cdots (\partial /\partial z_m)^{\alpha_m}$ and $<z> = \sqrt{1 + |z|^2}$. 
Let $L^2
=L^2(R^m)$ be the space of all square integrable functions in $R^m$ with 
inner product
$(\cdot,\cdot)$ and  norm $\Vert\cdot\Vert$.  We introduce the weighted 
Sobolev spaces
$B^a(R^m) := \{f \in L^2;\
\|f\|_{B^a} := \|f\| + \sum_{|\alpha| = a} (\|z^{\alpha}f\| +
\|(\hbar\partial_z)^{\alpha}f\|) < \infty\}\ (a = 1,2,\dotsc)$. Let 
$B^{-a}(R^m)$ denote its
dual space. We set
$B^0 : = L^2$.  Let
$\Sspace = \Sspace (R^m)$ be the Schwartz space of all rapidly decreasing 
functions in $R^m$.
  \par
    Let $\chi \in \Cspace (R^{m'})$ with compact support such that $\chi (0) 
= 1.$  For a
function
$g(z,z')$ in
$R^{m}\times R^{m'}$ we define the oscillatory integral 
$\text{Os}-{\displaystyle\int}
g(\cdot,z')dz'$ by
$\lim_{\epsilon
\rightarrow 0}{\displaystyle\int}\chi (\epsilon z') g(\cdot,z')dz'$ 
independently of
the choice of $\chi$ pointwise, in the topology of $B^a(R^{m})$ or in the 
topology in
$\Sspace(R^m)$ (cf.
\cite{Kumano-go}).
\section{Main results}

We arbitrarily cut off the terms   of  large wave numbers
$k$ in \eqref{2.22}.  That is, let $M_j\ (j=1,2,3)$ be  arbitrary positive
integers such that
$M_2 \leq M_3$.
We consider
\begin{align} \label{3.1}
\Lambda_j := \Bigl\{k = &
\left(\frac{2\pi}{L_1}s_1,\frac{2\pi}{L_2}s_2,
\frac{2\pi}{L_3} s_3 \right); s_1^2 + s_2^2 + s_3^2 \not= 0, \notag \\
    & |s_1|, |s_2|, |s_3| \leq M_j \Bigr\}.
\end{align}
Then we take $\Lambda'_j\ (j=1,2,3)$ such that
\begin{equation} \label{3.2}
\Lambda_j =: \Lambda'_j \cup -\Lambda'_j,\  \Lambda'_j \cap -\Lambda'_j=
\text{empty
set}, \ \Lambda'_2 \subseteq \Lambda'_3
\end{equation}
and fix $\Lambda'_j$ hereafter.
Let $N_j$ denote the number of elements of the set $\Lambda'_j$.
It follows from \eqref{2.13} that
$a_{\Lambda'_j} : = \{\alki\}_{k\in
\Lambda'_j, \mathrm{i,l}} \in R^{4N_j}$ are independent variables (cf. 
p.154 in \cite{Spohn}).
  \par
We  consider
\begin{align} \label{3.3}
&\tilde{{\cal L}_c}(\xvec,\dot{\xvec}, \{\alk\}, \{\dot{a}_{\mathrm{l}k}\})
:= \sum_{j = 1}^n
\frac{m_j }{2}|\dot{x}^{(j)}|^2   \notag \\
&  \quad - \frac{2\pi}{|V|}\sum_{k \in \Lambda_1}\sum_{j,l=1,j\not=l}^n
\frac{e_j e_l\cos k\cdot (x^{(j)} - x^{(l)})}{|k|^2 } \notag \\
& \quad + \frac{1}{ c}\sum_{j=1}^n  e_j\/\dot{x}^{(j)}\cdot
\tilde{A}(x^{(j)},a_{\Lambda'_2})  \notag \\
&  \quad + \frac{1}{2}\sum_{k \in \Lambda_3,\mathrm{i,l}}
\left(\frac{\left(\dot{a}_{\mathrm{l}k}^{(\mathrm{i})}\right)^2}{2|V|} -
\frac{(c|k|)^2\left(\alki\right)^2}{2|V|}  + \frac{\hbar c|k|}{2}\right)
\end{align}
   in place
of  ${\cal L}_c$ given by \eqref{2.22}, where $A$ given by (2.14)
is replaced with
\begin{align} \label{3.4}
& \tilde{A}(x,a_{\Lambda'_2}) = \frac{\sqrt{4\pi}}{|V|} cg(x) \sum_{k \in
\Lambda_2}
\sum_{\mathrm{l}=1}^2 \frac{1}{\sqrt{2}}\Big(
\psi ( a_{\mathrm{l}k}^{(1)}) \cos k\cdot x  \notag \\
&\quad \quad  + \psi ( a_{\mathrm{l}k}^{(2)}) \sin k\cdot x
\Bigr) \overrightarrow{e_{\mathrm{l}}} (k).
\end{align}
We assume $\psi (-\theta) = -\psi(\theta)\ (\theta \in R)$.  \par
   For the sake of simplicity we write $\Lambda' := \Lambda'_3$ and  $N := 
N_3$.  We consider
a subdivision
\begin{equation*}
   \Delta : 0 = \tau_0 < \tau_1 < \dotsc < \tau_{\nu} = T, \quad
|\Delta|: = \max_{1\leq l \leq \nu}(\tau_l -\tau_{l-1})
\end{equation*}
of $[0,T]$. Let $\xvec \in R^{3n}$ and $a_{\Lambda'} \in R^{4N}$ be fixed.
We take arbitrarily
\[
   \xvec^{(0)},\dotsc, \xvec^{(\nu -1)}  \in R^{3n}
   \]
   and
   \[
a_{\Lambda'}^{(0)},\dotsc, a_{\Lambda'}^{(\nu -1)} \in R^{4N}.
\]
  Then, we write the broken line path on $[0,T]$ connecting
$\xvec^{(l)}$ at $\theta = \tau_l\ (l = 0,1,\dotsc,\nu,\ \xvec^{(\nu)} =
\xvec)$ in order
as $ \qvecdelta (\theta) \in R^{3n}$.  Of course, $d\qvecdelta (\theta)/d\theta
=:\dot{\overrightarrow{q}}_{\Delta}(\theta)$ in the distribution sense is 
in $L^2([0,T])$. In
the same way we define the broken line path
$a_{\Lambda'\Delta} (\theta)
\in R^{4N}$ on $[0,T]$ for $a_{\Lambda'}^{(0)},\dotsc,
a_{\Lambda'}^{(\nu -1)}$ and $a_{\Lambda'}$.
We define $a_{\Lambda \Delta} (\theta) \in R^{8N}$ by means of \eqref{2.13}.
We write the classical action
\begin{align} \label{3.5}
& S_c(T,0;\qvecdelta, a_{\Lambda \Delta}) = \int_0^T\tilde{{\cal
L}}_c
\bigl(\qvecdelta (\theta), \dot{\qvec}_{\Delta}(\theta),
\notag \\
& \qquad \quad  a_{\Lambda \Delta} (\theta),
\dot{a}_{\Lambda \Delta}(\theta) \bigr)d\theta.
\end{align} \par
   Let $\rho^* > 0$ be the constant, which will be defined for $\Lambda'_1, 
\Lambda'_2$ and
$\Lambda'_3$ in Proposition 7.2 of the present paper.  See also Remark 7.1. 
Then we have
\begin{thm}
   We assume for $g(x)$ and $\psi (\theta)$ in
\eqref{3.4} that for any $ l = 1,2,\dotsc$ and any multi-index $\alpha$
there exist  constants
$\delta_l > 0$ and $\delta_{\alpha} > 0$ satisfying
\begin{equation} \label{3.6}
     |\partial_{\theta}^{l}\/\psi (\theta)|\leq C_{l}<\theta>^{-(1+ 
\delta_l)}, \
\theta \in R
\end{equation}
and
\begin{equation} \label{3.7}
|\partial_x^{\alpha}g(x)|\leq C_{\alpha}<x>^{-(1+ \delta_{\alpha})}, \  x 
\in R^3.
\end{equation}
Let $|\Delta| \leq \rho^*$ and $f(\xvec, a_{\Lambda'}) \in B^a(R^{3n+4N})\ 
(a = 0,1,2,\dotsc).$
Then, there exists the function
\begin{align} \label{3.8}
   & \left\{\prod_{l=1}^{\nu} \left(\prod_{j=1}^n  \sqrt{\frac{m_j }{2\pi
i\hbar(\tau_l -
\tau_{l-1})}}^{\raisebox{3pt}{\kern3pt$3$}} \right)
  \sqrt{\frac{1}{2\pi i|V|\hbar(\tau_l -
\tau_{l-1})}}^{\raisebox{3pt}{\kern3pt$4N$}} \right\}\notag \\
&  \times  \text{Os} - \int \cdots
\int \left(\exp
i\hbar^{-1}S_c(T,0;\qvec_{\Delta},a_{\Lambda \Delta})\right)
f\Bigl(\qvec_{\Delta} (0),\notag \\
& a_{\Lambda'\Delta}(0)\Bigr) d\xvec^{(0)}\cdots d\xvec^{(\nu-1)}
da_{\Lambda'}^{(0)}\cdots da_{\Lambda'}^{(\nu-1)}
\end{align}
in $B^a(R^{3n+4N})$, which we write
as
$\left(C_{\Delta}(T,0)f\right)(\xvec,a_{\Lambda'})$ or
$\displaystyle{\iint} \bigl(\exp
i\hbar^{-1}
S_c(T,0;\\
\qvecdelta,a_{\Lambda \Delta}) \bigr)
f\left(\qvecdelta(0),a_{\Lambda'\Delta}(0)\right)  {\cal D}\qvecdelta  {\cal
D}a_{\Lambda'\Delta}$.
In addition,
as
$|\Delta|$ tends to $0$, the function $\left(\Cdelta (T,0)f\right)(\xvec,
a_{\Lambda'})$ converges to the so-called Feynman path integral
$\displaystyle{\iint}
\left(\exp i\hbar^{-1}
S_c(T,0;\qvec,a_{\Lambda} \right)
f (
\qvec(0),a_{\Lambda'}(0)) {\cal D}\qvec  {\cal D}a_{\Lambda'}$
in $B^a(R^{3n+4N})$.  We also see that this limit, which is $B^a$-valued 
continuous and
$B^{a-2}$-valued continuously differentiable in $T \in (0,\infty)$, 
satisfies the Schr\"odinger
type equation
\begin{equation} \label{3.9}
        i\hbar \frac{\partial }{\partial t}u(t)  = H(t)u(t)
\end{equation}
with $u(0) = f$,
where
\begin{align} \label{3.10}
       &  H(t) =  \sum_{j=1}^n \frac{1}{2m_j }
      \left|\frac{\hbar}{i}\frac{\partial}{\partial x^{(j)}} - \frac{e_j}{c}
\tilde{A}(x^{(j)},a_{\Lambda'_2}) \right|^2  \notag \\
&\hspace{1.5cm}  + \frac{2\pi}{|V|}\sum_{k \in \Lambda_1}\sum_{j,l=1,j\not= 
l}^n
\frac{e_j e_l\cos k\cdot (x^{(j)} - x^{(l)})}{|k|^2} \notag \\
& + \sum_{k \in \Lambda',\mathrm{i,l}} \left\{
\frac{|V|}{2} \left(\frac{\hbar}{i} \frac{\partial}{\partial
\alki}\right)^2 +\frac{(c|k|)^2}{2|V|} \left(\alki \right)^2 -
\frac{\hbar c|k|}{2}
\right\}.
\end{align}
\end{thm}
\begin{rem}
We determine an indefinite constant $C$ in \eqref{2.3} by
\begin{equation*}
  \frac{2\pi}{|V|} \sum_{j=1}^n e_j^2 \sum_{k\in \Lambda_1}\frac{1}{|k|^2} +
\frac{1}{2}\sum_{k\in \Lambda_3}\frac{\hbar c |k|}{2}
\end{equation*}
and cut off the terms of large wave numbers $k$ of \eqref{2.19} by
introducing  $\Lambda_j\ (j = 1,2,3)$. Then we get \eqref{3.3} again, 
taking the account of
the constraints \eqref{2.19'}.
\end{rem}
\begin{rem}
Let $0 \leq t_0 \leq t \leq T$.  For $f \in B^a(R^{3n+4N})\ (a = 
0,1,2,\dotsc)$  we define
$\Cdelta (t,t_0)f$ as in \eqref{3.8}. See \eqref{9.3} in the present paper 
for the precise
definition. As will be seen in the proof of Theorem 3.1, under the 
assumptions of Theorem 3.1
there exist
$\left(\Cdelta (t,t_0)f\right)(\xvec,a_{\Lambda'})$ in
$B^a$ and
$\lim_{|\Delta| \rightarrow 0}\Cdelta (t,t_0)f$ in $B^a$ uniformly in $0 
\leq t_0 \leq t \leq
T$, which satisfies the Sch\"odinger type equation \eqref{3.9} with $u(t_0) 
= f$.
\end{rem}
\par
In place of ${\cal L}$ expressed by \eqref{2.19}  we consider
\begin{align} \label{3.11}
&\tilde{{\cal L}}(\xvec,\dot{\xvec}, \{\alk\},
\{\dot{a}_{\mathrm{l}k}\},\{\phi_k\}) :=
\sum_{j=1}^n
\frac{m_j }{2}|\dot{x}^{(j)}|^2   \notag \\
& \quad + \frac{1}{8\pi |V|}\sum_{k \in \Lambda_1}
\Biggl\{\sum_{\mathrm{i}=1}^2\left(|k|^2 
\left(\phi^{(\mathrm{i})}_k\right)^2  -
8\pi
\rho^{(\mathrm{i})}_k(\xvec)\ \phiki \right)    \notag \\
& \quad + 16\pi^2
\frac{\sum_{j=1}^n e_j^2}{|k|^2}
\Biggr\} + \frac{1}{c}\sum_{j=1}^n e_j\ \dot{x}^{(j)}\cdot
\tilde{A}(x^{(j)},a_{\Lambda'_2})  \notag \\
&  \quad + \frac{1}{ 2}\sum_{k\in \Lambda_3,\mathrm{i,l}}
\left(\frac{\left(\dot{a}_{\mathrm{l}k}^{(\mathrm{i})}\right)^2}{2|V|} -
\frac{(c|k|)^2\left(\alki\right)^2}{2|V|}  + \frac{\hbar c|k|}{2}\right)
\end{align}
by means of \eqref{3.4}  as in $\tilde{{\cal
L}}_c$.
   \par
    Let $\qvecdelta (\theta) \in R^{3n}$,
    $a_{\Lambda'\Delta} (\theta) \in R^{4N}$ and $a_{\Lambda \Delta} (\theta)
\in R^{8N}$
     be the broken line paths defined
before. Let $\xivec_k := \left\{\xi_{k}^{(\mathrm{i})}\right\}_{\mathrm{i} 
= 1,2} \in R^{2}$
for
$k
\in \Lambda'_1$. Take $\xivec_k^{(0)},
\xivec_k^{(1)}, \dotsc$ and $
\xivec_k^{(\nu - 1)}$ in $R^2$ arbitrarily. Set $\rho_k (\xvec) :=
(\rho_k^{(1)}(\xvec),\rho_k^{(2)}(\xvec))$
by means of \eqref{2.16} and (2.17). Then, we define
     the path
    \begin{equation} \label{3.12}
    \phi_{k\Delta}(\theta) := \xivec_k^{(l)} +
\frac{4\pi\rho_k(\qvecdelta(\theta))}{|k|^2} \in R^2, \ \tau_{l-1} < \theta
\leq \tau_l
\end{equation}
$(l = 1,2,\dotsc,\nu)$, where $\phi_{k\Delta}(0) :=
    \lim_{\theta \rightarrow 0+0}\phi_{k\Delta}(\theta)$. We set
$\phi_{\Lambda'_1\Delta}(\theta) := \{\phi_{k\Delta}(\theta)\}_{k \in
\Lambda'_1}\\
\in R^{2N_1}$.  We define $\phi_{\Lambda_1 \Delta}(\theta) \in R^{4N_1}$ by 
means
of (2.13).
Let $S(T,0;\qvecdelta,a_{\Lambda \Delta},\phi_{\Lambda_1 \Delta})$ be
the classical
action for $\tilde{{\cal L}}(\xvec,\dot{\xvec}, \{\alk\},
\{\dot{a}_{\mathrm{l}k}\},
\{\phi_k\})$ given by \eqref{3.11}.
\begin{thm}
    Let $|\Delta| \leq \rho^*$ and $f(\xvec,a_{\Lambda'}) \in 
B^a(R^{3n+4N})\ (a =
0,1,2,\dotsc)$.  Then, under the assumption of Theorem 3.1 there exists the 
function
\begin{align} \label{3.13}
      &
    \Biggl[ \prod_{l=1}^{\nu} \Biggl\{\left(\prod_{j=1}^n  \sqrt{\frac{m_j 
}{2\pi
i\hbar(\tau_l -
\tau_{l-1})}}^{\raisebox{5pt}{\kern3pt$3$}} \right)
  \sqrt{\frac{1}{2\pi i\hbar|V|(\tau_l -
\tau_{l-1})}}^{\raisebox{5pt}{\kern3pt$4N$}} \notag \\
&\times  \prod_{k \in \Lambda'_1}  \frac{|k|^2
(\tau_l - \tau_{l-1})}{4 i\pi^2\hbar|V|}  \Biggr\}
\Biggr] \text{Os} - \int
\dotsi\int \Bigl(\exp
i\hbar^{-1} S(T,0;\notag \\
&  \qvec_{\Delta},a_{\Lambda \Delta}, \phi_{\Lambda_1 \Delta})
\Bigr)f\left(\qvec_{\Delta} (0),a_{\Lambda'\Delta}(0)\right)
            d\xvec^{(0)}\cdots d\xvec^{(\nu-1)}
    \notag \\
& \cdot da_{\Lambda'}^{(0)}\cdots
da_{\Lambda'}^{(\nu-1)} \prod_{k \in \Lambda'_1}  d\xivec_k^{(0)}
d\xivec_k^{(1)}
\cdots
d\xivec_k^{(\nu-1)}
\end{align}
in $B^a(R^{3n+4N})$, which is equal to \begin{align*}
&
\iint \left(\exp i\hbar^{-1}
S_c(T,0;\qvecdelta,a_{\Lambda \Delta}) \right) \notag \\
& \qquad \times f
\left(\qvecdelta(0),a_{\Lambda'\Delta}(0)\right){\cal D}\qvecdelta
{\cal D}a_{\Lambda'\Delta}
\end{align*}
defined by \eqref{3.8} in Theorem 3.1.
So it follows from Theorem 3.1 that  as $|\Delta| \rightarrow 0$, then
\eqref{3.13} converges to the Feynman path integral 
$\left(G(T,0)f\right)(\xvec,a_{\Lambda'})$
or
   \begin{equation} \label{3.14}
   \iiint \left(\exp i\hbar^{-1}
S(T,0;\qvec,a_{\Lambda},\phi_{\Lambda_1})\right)
f
\left(\qvec(0),a_{\Lambda'}(0)\right){\cal D}\qvec
   {\cal D}a_{\Lambda'}{\cal D}\phi_{\Lambda'_1}
   \end{equation}
in $B^a(R^{3n+4N})$, which satisfies the Schr\"odinger type equation 
\eqref{3.9} with $u(0) =
f$.
   This expression \eqref{3.14} is given in \S 9-8 in Feynman-Hibbs
\cite{Feynman-Hibbs}
without any comments on its  definition.
   \end{thm}
\begin{rem}  As was noted in the introduction, the  constraints
(2.20) are not
needed in Theorem 3.2 above. The path $\phi_{k\Delta}(\theta)$ defined by 
(3.12) is
determined so that $\partial \tilde{{\cal
L}}(\overrightarrow{q}_{\Delta}(\theta),\dot{\overrightarrow{q}}_{\Delta}
(\theta),
a_{\Lambda\Delta}(\theta),\dot{a}_{\Lambda\Delta}(\theta),\phi_{\Lambda_1\Delta}(\theta))
/\partial
\phi_{k}^{(\mathrm{i})}\ (\mathrm{i} = 1,2)$ is piecewise constant.
\end{rem}
\begin{rem}
We take $f \in {\cal S}(R^{3n + 4N})$ and set $M_0 = [(3n+4N)/2] +1$, where 
$[\cdot]$ denotes
Gauss' symbol. Let $\zeta = (x,X)$, and $\alpha$ and $\beta$ multi-indices. 
Then, the Sobolev
inequality  shows
\begin{equation*}
  \sup_{\zeta \in R^{3n+4N}} 
\left|\zeta^{\alpha}\partial_{\zeta}^{\beta}f(\zeta) \right|
\leq \left\Vert \zeta^{\alpha}\partial_{\zeta}^{\beta}f\right\Vert +
\sum_{|\kappa| = M_0} \left\Vert
\partial_{\zeta}^{\kappa}\left( 
\zeta^{\alpha}\partial_{\zeta}^{\beta}f\right)\right\Vert.
\end{equation*}
\end{rem}
It follows from Lemma 2.4 with $a = b = 1$ in \cite{Ichinose 1995} or as in 
the proof of
\eqref{7.14} in the present paper that the rhs of the above is bounded by 
$C_{\alpha,
\beta}\Vert f
\Vert_{B^{|\alpha + \beta| + M_0}}$ with a constant $C_{\alpha,\beta}$. 
Hence, for $|\Delta|
\leq \rho^*$ there exist \eqref{3.8}, \eqref{3.13}, the limit of 
\eqref{3.8} as $|\Delta|
\rightarrow 0$ and the limit of \eqref{3.13} as $|\Delta| \rightarrow 0$ in 
the topology of
$\Sspace$, so pointwise.
\begin{rem}
Let $0 \leq t_0 \leq t \leq T$.  For $f \in B^a(R^{3n+4N})\ (a = 
0,1,2,\dotsc)$  we can define
$G_{\Delta} (t,t_0)f$ as in \eqref{3.13} in the same way that $\Cdelta 
(t,t_0)f$ is defined
in Remark 3.2. See also \eqref{9.20} in the present paper. As will be seen 
in the proof of
Theorem 3.2, under the  assumptions of Theorem 3.1 there exists
$G_{\Delta} (t,t_0)f$ in
$B^a$ and $G_{\Delta} (t,t_0)f$ is equal to $\Cdelta (t,t_0)f$.
\end{rem}
\par
We consider an external electromagnetic field $E_{\text{ex}}(t,x) =
(E_{\text{ex}1},
E_{\text{ex}2}, E_{\text{ex}3}) \in R^3$ and
$B_{\text{ex}}(t,x)  =
(B_{\text{ex}1},B_{\text{ex}2}, B_{\text{ex}3})\in R^3$ such that
$\partial_x^{\alpha}E_{\text{ex}j}(t,x), 
\partial_x^{\alpha}B_{\text{ex}j}(t,x)$ and
$\partial_tB_{\text{ex}j}(t,x)
\ (j = 1,2,3)$  are continuous in
$[0,T]\times R^n$ for all $\alpha$.
Let $\phi_{\text{ex}}(t,x)  \in R$ and
$A_{\text{ex}}(t,x)
\in R^3$ be the electromagnetic potentials to $E_{\text{ex}}$ and
$B_{\text{ex}}$.  Then we get Theorem 3.3 below.  Though Theorem 3.3 gives 
the generalization
of Theorems 3.1 and 3.2, the results are stated separately from Theorems 
3.1 and 3.2 to avoid
confusion.
\par
  We
replace
$\tilde{A}(x^{(j)},a_{\Lambda'_2})$ in \eqref{3.3}, \eqref{3.10} and 
\eqref{3.11} with
$\tilde{A}(x^{(j)},a_{\Lambda'_2}) + A_{\text{ex}}(t,x^{(j)})$. Moreover we 
add $- \sum_{j=1}^n
e_j\phi_{\text{ex}}(t,x^{(j)})$ to \eqref{3.3} and \eqref{3.11}, and 
$\sum_{j=1}^n
e_j\phi_{\text{ex}}(t,x^{(j)})$ to \eqref{3.10}, respectively.   Then we have
\par
\begin{thm}  Besides the assumptions of Theorem 3.1
   we suppose as in
Ichinose \cite{Ichinose 1999,Ichinose 2000,Ichinose 2003} that for any 
$\alpha \not= 0$ there
exist constants $C_{\alpha}$
and $\delta_{\alpha} > 0$ satisfying
\begin{equation} \label{3.15}
   |\partial_x^{\alpha}E_{\text{ex}\hspace{0.08cm} j}(t,x)| \leq C_{\alpha},\
|\partial_x^{\alpha}B_{\text{ex}\hspace{0.08cm} j}(t,x)| \leq 
C_{\alpha}<x>^{-(1 +
\delta_{\alpha})}
\end{equation}
  and
\begin{equation} \label{3.16}
|\partial_x^{\alpha}A_{\text{ex}\hspace{0.08cm} j}(t,x)| \leq C_{\alpha},\
|\partial_x^{\alpha}\phi_{\text{ex}}(t,x)| \leq C_{\alpha}<x>
\end{equation}
for j = 1, 2 and  3 in $[0,T]\times R^n$.    Then, the
same assertions as in
Theorems 3.1 and 3.2 hold.
\end{thm}
\begin{rem}
It follows from Lemma 6.1 in \cite{Ichinose 1999} that under the 
assumptions \eqref{3.15}
there exist
$A_{\text{ex}}$ and $\phi_{\text{ex}}$ satisfying \eqref{3.16}.
\end{rem}
\section{The appearance of the Coulomb potentials}
   We will show rigorously that the Coulomb potentials appear as the limit 
of the second term
on the rhs of \eqref{3.3} and the limit of the second term
on the rhs of \eqref{3.10}.  This result is well known in physics (cf.
\cite{Fermi,Feynman-Hibbs}). We will give the rigorous proof.  Our proof is 
somewhat delicate.
\begin{thm}
   Let $L_j\ (j = 1,2,3)$ tend to $\infty$ under
the condition
\begin{equation} \label{4.1}
\lim_{L_1,L_2,L_3 \rightarrow \infty} \frac{L_i}{L_jL_k} = 0, \ (i,j,k) = 
(1,2,3),
(2,3,1), (3,1,2).
\end{equation}
Then we have
\begin{align} \label{4.2}
& \lim_{L_1,L_2,L_3 \rightarrow \infty} \lim_{M_1 \rightarrow \infty}
\frac{2\pi}{|V|}
\sum_{k \in \Lambda_1}\sum_{j,l =1,j \not= l}^n \frac{e_j e_l\cos k\cdot 
(x^{(j)} -
x^{(l)})}{|k|^2}
\notag
\\
& = \lim_{M_1 \rightarrow \infty}\lim_{L_1,L_2,L_3 \rightarrow \infty}
\frac{2\pi}{|V|}
\sum_{k \in \Lambda_1}\sum_{j,l =1,j \not= l}^n \frac{e_j e_l\cos k\cdot 
(x^{(j)} -
x^{(l)})}{|k|^2}
\notag
\\
& = \frac{1}{2} \sum_{j,l=1,j \not= l}^n \frac{e_j e_l}{\bigl| x^{(j)} -
x^{(l)}\bigr|} \quad \text{in}\ {\cal S'}(R^{3n}).
\end{align}
\end{thm}
    Let $\chi_0(k)$ be the function in $R^3$ defined by
\begin{equation}  \label{4.3}
\chi_0(k) :=
        \begin{cases}
            \begin{split}
              & 1,
                      \end{split}
         & |k| \leq 1 ,
                  \\
\begin{split}
        & 0,
             \end{split}
            & |k| >1.
        \end{cases}
\end{equation}
We first prove

\begin{lem}
Let $\epsilon > 0$.  Then we have
\begin{align} \label{4.4}
& \lim_{\epsilon \rightarrow 0}\frac{1}{(2\pi)^2}\sum_{j,l =1,j \not= l}^n
e_j e_l \int \frac{\cos k\cdot (x^{(j)} - x^{(l)})}{|k|^2}\chi_0(\epsilon 
k)dk \notag \\
& = \frac{1}{2} \sum_{j,l=1,j \not= l}^n \frac{e_j e_l}{\bigl| x^{(j)} -
x^{(l)}\bigr|} \quad \text{in}\ {\cal S'}(R^{3n}).
\end{align}
\end{lem}
\begin{proof}
Let $x$ and $k$ be in $R^3$.  Then, it is well known that
\begin{equation} \label{4.5}
\frac{1}{(2\pi)^2} \int \frac{\mathrm{e}^{ik\cdot x}}{|k|^2}dk = \frac{1}{2|x|}
\quad \text{in}\ {\cal S'}(R^{3})
\end{equation} 
(cf. \S 5.9 in \cite{Lieb-Loss}).
\par
For the sake of simplicity we consider the case of $n = 2$.  Let $x = 
x^{(1)}$ and $y =
x^{(2)}$.  We will prove
\begin{align} \label{4.6}
& \lim_{\epsilon \rightarrow 0}\frac{1}{(2\pi)^2} \int 
\frac{\mathrm{e}^{ik\cdot
(x - y)}}{|k|^2}\chi_0(\epsilon k)dk
   =   \frac{1}{2|x - y|} \quad \text{in}\ {\cal S'}(R^{6}).
\end{align}
We write the operation as $<T,f>$ for $T \in \Sspace'$ and $f \in \Sspace$.
Let $\varphi (x,y) \in \Sspace(R^6)$.  Then from \eqref{4.6} we have
\begin{align*}
& \lim_{\epsilon \rightarrow 0}\left<\frac{1}{(2\pi)^2} \int 
\frac{\mathrm{e}^{ik\cdot
(x - y)}}{|k|^2}\chi_0(\epsilon k)dk, \varphi (x,y)\right>  \notag \\
& = \lim_{\epsilon \rightarrow 0} \Biggl\{\left<\frac{1}{(2\pi)^2} \int 
\frac{\cos k\cdot (x -
y)}{|k|^2}\chi_0(\epsilon k) dk,  \varphi (x,y)\right>   \\
&\qquad  + i \left<\frac{1}{(2\pi)^2} \int
\frac{\sin k\cdot (x - y)}{|k|^2}\chi_0(\epsilon k) dk, \varphi 
(x,y)\right>\Biggr\}
\notag \\
& = \lim_{\epsilon \rightarrow 0} \left<\frac{1}{(2\pi)^2} \int \frac{\cos 
k\cdot (x -
y)}{|k|^2}\chi_0(\epsilon k) dk,  \varphi (x,y)\right>   \\
&  =   \left<\frac{1}{2|x - y|}, \varphi (x,y)\right>.
\end{align*}
Consequently we obtain \eqref{4.4}. \par
    The equation
\eqref{4.6} is equivalent to
\begin{align} \label{4.7}
& \lim_{\epsilon \rightarrow 0}\left<\frac{1}{2\pi^2} \int 
\frac{\mathrm{e}^{ik\cdot
(x - y)}}{|k|^2}\chi_0(\epsilon k)dk, \varphi (x,y)\right>  \notag \\
&  = \iint   \frac{1}{|x - y|} \varphi (x,y)dxdy
\end{align}
for all $\varphi (x,y) \in \Sspace(R^6)$. We set $x' = (x - y)/\sqrt{2}$ 
and $y' = (x +
y)/\sqrt{2}$.
    Let $\psi_1(x')$ and $\psi_2(y')$ be in $\Sspace(R^3)$.  We take 
$\varphi(x,y) =
\tilde{\varphi}(x',y') := \psi_1(x')\psi_2(y')$ in the lhs of \eqref{4.7}. 
Then the lhs of
\eqref{4.7} is equal to
\begin{align*}
& \lim_{\epsilon \rightarrow 0}\frac{1}{2\pi^2} \iiint 
\frac{\mathrm{e}^{ik\cdot
(x - y)}}{|k|^2}\chi_0(\epsilon k)\psi_1(x')\psi_2(y')dkdx'dy'  \notag \\
& = \lim_{\epsilon \rightarrow 0} \frac{1}{2\pi^2} \int \psi_1(x') dx'
\int\frac{\mathrm{e}^{ik\cdot
\sqrt{2}x'}}{|k|^2} \chi_0(\epsilon k)dk \int \psi_2(y')dy',
\end{align*}
which is also equal to
\begin{equation*}
\int \frac{1}{\sqrt{2}|x'|}\psi_1(x') dx'
  \int \psi_2(y')dy' = \iint \frac{\tilde{\varphi}(x',y')}{\sqrt{2}|x'|}dx'dy'
= \iint \frac{\varphi (x,y)}{|x - y|}dxdy
\end{equation*}
from \eqref{4.5}.  So, \eqref{4.7} holds for $\varphi (x,y) =
\psi_1(x')\psi_2(y')$.  Since the set of all linear combinations of 
$\psi_1(x')\psi_2(y')$ for
all $\psi_1$ and $\psi_2$ in $\Sspace(R^3)$ is dense in 
$\Sspace(R^6_{x',y'})$, so
\eqref{4.7} holds for all $\varphi(x,y) \in \Sspace(R^6)$.  Hence we get 
\eqref{4.6}.
\end{proof}
\begin{pro}
Let $c \geq 0$ be a constant.  Let $\Phi(k)$ be continuous in
$R^3\backslash(\{0\}\cup\{|k|=c\})$.
   We suppose $|\Phi(k)| \leq \phi(|k|)\ (k \in R^3)$, where
$\phi(r)$ is non-increasing in $(0,\infty)$, and $r^2\phi(r)$ is in 
$L^1([0,\infty))$ and
bounded in $(0,\infty)$.  Then, $\left((2\pi)^3/|V|\right)\sum_{k \not= 
0}\Phi(k)$ is
convergent absolutely, where the sum of $k$  is taken over $\left(2\pi 
s_1/L_1,2\pi
s_2/L_2, 2\pi s_3/L_3 \right)
   \ (s_1,s_2,  s_3  = 0, \pm 1, \pm 2, \dotsc)$. We also get
\begin{equation}  \label{4.8}
\lim_{L_1,L_2,L_3 \rightarrow \infty} \frac{(2\pi)^3}{|V|}\sum_{k \not= 
0}\Phi(k)
= \int \Phi(k)dk
\end{equation}
under the condition \eqref{4.1}.
\end{pro}
\begin{proof}
   We write $L =
(L_1,L_2,L_3).$  Let's define the step function $\Phi_L (k)$ by
\begin{align*}
& \Phi_L (k) =  \Phi \left(\frac{2\pi s_1}{L_1},\frac{2\pi 
s_2}{L_2},\frac{2\pi s_3}{L_3}
\right), \quad k \in \Bigl(\frac{2\pi (s_1 -1)}{L_1},\frac{2\pi 
s_1}{L_1}\Bigr] \\
&\qquad \times \Bigl(\frac{2\pi (s_2 -1)}{L_2},\frac{2\pi s_2}{L_2}\Bigr] 
\times
\Bigl(\frac{2\pi (s_3 -1)}{L_3},\frac{2\pi s_3}{L_3}\Bigr], \\
& \Phi_L (k) =  \Phi \left(\frac{2\pi s_1}{L_1},-\frac{2\pi 
s_2}{L_2},\frac{2\pi s_3}{L_3}
\right), \quad k \in \Bigl(\frac{2\pi (s_1 -1)}{L_1},\frac{2\pi 
s_1}{L_1}\Bigr] \\
&\qquad \times \Bigl[-\frac{2\pi s_2}{L_2},-\frac{2\pi (s_2 -1)}{L_2}\Bigr) 
\times
\Bigl(\frac{2\pi (s_3 -1)}{L_3},\frac{2\pi s_3}{L_3}\Bigr]
\end{align*}
for $s_1, s_2, s_3 = 1,2,\dotsc$.  Then, for $k \in \bigl(2\pi (s_1 -1)/L_1,
2\pi s_1/L_1\bigr] \times \bigl(2\pi (s_2 -1)/L_2,2\pi s_2/L_2\bigr]
\times
\bigl(2\pi (s_3 -1)/L_3,2\pi s_3/L_3\bigr]$ we have
\begin{align*}
|\Phi_L(k)| & = \left|\Phi \left(\frac{2\pi s_1}{L_1},\frac{2\pi 
s_2}{L_2},\frac{2\pi s_3}{L_3}
\right)\right| \\
& \leq \phi \left(\left| \left(\frac{2\pi s_1}{L_1},\frac{2\pi 
s_2}{L_2},\frac{2\pi s_3}{L_3}
\right)\right|\right) \leq \phi(|k|)
\end{align*}
since $\phi(r)$ is non-increasing.  In the same way, for $k \in \bigl(2\pi 
(s_1 -1)/L_1,
2\pi s_1/L_1\bigr] \times \bigl[-2\pi s_2/L_2, -2\pi (s_2 -1)/L_2\bigr)
\times
\bigl(2\pi (s_3 -1)/L_3,2\pi s_3/L_3\bigr]$ we get
\begin{equation} \label{4.9}
|\Phi_L(k)| \leq \phi(|k|).
\end{equation}
In the same way as in the above we can define the step function $\Phi_L(k)$ 
for all
$k \in R^3\backslash\{0\}$
such that \eqref{4.9} and \eqref{4.10} below hold. It holds that
\begin{equation}  \label{4.10}
  \frac{(2\pi)^3}{|V|}\sum_{k \not= 0}\Phi(k)
= \int_{R^3} \Phi_L(k)dk + \frac{(2\pi)^3}{|V|}\sum_{k \not= 0, s_1s_2s_3 = 
0}\Phi(k).
\end{equation} \par
   For a short while we suppose $L_1 \leq L_2 \leq L_3$.  Since $\phi(r)$ is 
non-increasing, it
holds that for $s_1 \geq 2$  we have
\begin{align*}
&\phi \left(\left| \left(\frac{2\pi s_1}{L_1},0,0
\right)\right|\right)
\leq \phi \left(\left| \left(\frac{2\pi (s_1 
-1)}{L_1},\frac{2\pi}{L_2},\frac{2\pi}{L_3}
\right)\right|\right) \leq \phi(|k|), \\
&\qquad  k \in \Bigl(\frac{2\pi (s_1 -2)}{L_1},\frac{2\pi (s_1-1)}{L_1}\Bigr]
  \times \Bigl(0,\frac{2\pi}{L_2}\Bigr] \times
\Bigl(0,\frac{2\pi}{L_3}\Bigr]
\end{align*}
and also for $s_1 \geq 2$ and $s_2 \geq 1$
\begin{align*}
&\phi \left(\left| \left(\frac{2\pi s_1}{L_1},\frac{2\pi s_2}{L_2},0
\right)\right|\right)
\leq \phi \left(\left| \left(\frac{2\pi (s_1 -1)}{L_1},\frac{2\pi 
s_2}{L_2},\frac{2\pi}{L_3}
\right)\right|\right) \leq \phi(|k|), \\
&\qquad  k \in \Bigl(\frac{2\pi (s_1 -2)}{L_1},\frac{2\pi (s_1-1)}{L_1}\Bigr]
  \times \Bigl(\frac{2\pi (s_2-1)}{L_2},\frac{2\pi s_2}{L_2}\Bigr] \times
\Bigl(0,\frac{2\pi}{L_3}\Bigr].
\end{align*}
For $s_2 \geq 2$ we also have
\begin{align*}
&\phi \left(\left| \left(\frac{2\pi}{L_1},\frac{2\pi s_2}{L_2},0
\right)\right|\right)
\leq \phi \left(\left| \left(\frac{2\pi}{L_1},\frac{2\pi (s_2 
-1)}{L_2},\frac{2\pi}{L_3}
\right)\right|\right) \leq \phi(|k|), \\
&\qquad  k \in \Bigl(0,\frac{2\pi}{L_1}\Bigr]
  \times \Bigl(\frac{2\pi (s_2-2)}{L_2},\frac{2\pi (s_2-1)}{L_2}\Bigr] \times
\Bigl(0,\frac{2\pi}{L_3}\Bigr].
\end{align*}
Consequently we get
\begin{align}  \label{4.11}
  &\frac{(2\pi)^3}{|V|}\sum_{k \not= 0, s_3 =0}|\Phi (k)|
\leq \frac{(2\pi)^3}{|V|}\sum_{k \not= 0, s_3 = 0}\phi(|k|) \leq
\frac{(2\pi)^3}{|V|}\sum_{k \not= 0, s_3 = 0,s_1,s_2 = 0,\pm 1}\phi(|k|) 
\notag \\
&\qquad  + \frac{(2\pi)^3}{|V|}\sum_{k \not= 0, s_3 = s_1 =0}\phi(|k|)
+ 3\int_{0\leq k_3 \leq (2\pi)/L_3}\phi(|k|)dk.
\end{align}
We take a constant $m \geq 1$ such that $L_3 \geq mL_1 \geq L_2$.  We add 
the refinement
$\bigl\{\bigl((2\pi)/(mL_1),(2\pi s_2)/L_2,(2\pi s_3)/L_3 \bigr); 
s_2,s_3 = 0,\pm 1, \pm
2,\dotsc \bigr\}$ to $\bigl\{\bigl((2\pi s_1)/L_1, \\
(2\pi s_2)/L_2,(2\pi s_3)/L_3 \bigr); s_1,s_2,s_3= 0,\pm 1, \pm 2,\dotsc 
\bigr\}$.  Then,
for $s_2 \geq 2$ noting
\[
\phi \left(\left| \left(0,\frac{2\pi s_2}{L_2},0\right)\right|\right)
\leq \phi \left(\left| \left(\frac{2\pi}{mL_1},\frac{2\pi(s_2 - 
1)}{L_2},\frac{2\pi}{L_3}
\right)\right|\right),
\]
as in the proof of \eqref{4.11} we have
\begin{align*}
\frac{(2\pi)^3}{m|V|}\sum_{k \not= 0, s_3= s_1 =0}\phi (|k|)
  &\leq
\int_{0\leq k_1 \leq (2\pi)/(mL_1), 0\leq k_3 \leq (2\pi)/L_3}\phi(|k|)dk  \\
& \leq \frac{1}{m}\int_{0\leq k_3 \leq (2\pi)/L_3}\phi(|k|)dk.
\end{align*}
Consequently from \eqref{4.11} we get
\begin{align}  \label{4.12}
  &\frac{(2\pi)^3}{|V|}\sum_{k \not= 0, s_3 =0}|\Phi (k)|
  \leq
\frac{(2\pi)^3}{|V|}\sum_{k \not= 0, s_3 = 0,s_1,s_2 = 0,\pm 1}\phi(|k|) 
\notag \\
&\qquad
+ 4\int_{0\leq k_3 \leq (2\pi)/L_3}\phi(|k|)dk.
\end{align}
\par
    Consider the case of $L_1 \leq L_2$ generally.  We add the refinement
$\bigl\{\bigl((2\pi s_1)/L_1,\\
(2\pi s_2)/L_2,(2\pi)/(mL_3) \bigr); s_1,s_2 = 0,\pm 1, \pm
2,\dotsc \bigr\}$ to $\bigl\{\bigl((2\pi s_1)/L_1,
(2\pi s_2)/L_2,\\
(2\pi s_3)/L_3 \bigr); s_1,s_2,s_3 = 0,\pm 1, \pm 2,\dotsc \bigr\}$,
where  $m \geq 1$ is a constant such that $L_2 \leq mL_3$.  Then, as in
the proof of \eqref{4.12} we can also prove \eqref{4.12} for this case. 
Hence, for general
$L_1,L_2$ and $L_3$ we obtain
\begin{align}  \label{4.13}
  &\frac{(2\pi)^3}{|V|}\sum_{k \not= 0, s_j =0}|\Phi (k)|
  \leq
\frac{(2\pi)^3}{|V|}\sum_{k \not= 0,s_1,s_2,s_3 = 0,\pm 1}\phi(|k|) \notag \\
&\qquad
+ 4\int_{0\leq k_j \leq (2\pi)/L_j}\phi(|k|)dk\quad (j = 1,2,3).
\end{align}
\par
   From \eqref{4.9}, \eqref{4.10} and \eqref{4.13} we can prove that
$\sum_{k \not= 0}|\Phi (k)|$ is convergent.  It follows from the 
assumptions that
\[
\phi(|k|) \leq \text{Const.} \frac{1}{|k|^2}, \quad k \not= 0.
\]
So we see that $\bigl((2\pi)^3/|V|\bigr)\sum_{k \not= 0,s_1,s_2,s_3 = 0,\pm 
1}\phi(|k|)$ tends
to zero as $L_1, L_2$ and $L_3$ tend to the infinity under the condition 
\eqref{4.1}. We note
that we first used the condition \eqref{4.1} here. Consequently, from 
\eqref{4.13} and the
assumptions we have
\[
\lim_{L_1,L_2,L_3 \rightarrow \infty}\frac{(2\pi)^3}{|V|}\sum_{k \not= 
0,s_j = 0}\Phi(k) = 0,
\quad j = 1,2,3
\]
under \eqref{4.1}.  Hence, noting \eqref{4.9}, from \eqref{4.10} we obtain 
\eqref{4.8} by
means of the Lebesgue dominated convergence theorem.
\end{proof}
\begin{rem}
The condition \eqref{4.1} is necessary for $\Phi(k)$ in Proposition 4.3 to 
satisfy
\eqref{4.8} in general.  In fact, for example, suppose 
$\overline{\lim}_{L_1,L_2,L_3
\rightarrow
\infty} L_1/(L_2L_3) > 0$.  Let
$0
\leq
\chi_1(k)
\in
\Cspace (R^3)$ with compact support such that
$\chi_1(k) = 1\ (|k| \leq 1)$.  We consider $\Phi(k) = |k|^{-2}\chi_1(k)\ 
( \geq 0)$.  Then,
from \eqref{4.9} and \eqref{4.10}  we have
\begin{align*}
\overline{\lim}_{L_1,L_2,L_3
\rightarrow
\infty} \frac{(2\pi)^3}{|V|}\sum_{k \not= 0}\Phi (k) & \geq \int \Phi(k)dk +
\overline{\lim}_{L_1,L_2,L_3
\rightarrow
\infty} \frac{(2\pi)^3}{|V|}\Phi\left(\frac{2\pi}{L_1},0,0\right) \\
& > \int \Phi(k)dk.
\end{align*}
So, \eqref{4.8} doesn't hold.
\end{rem}
    \par
Now we will prove Theorem 4.1. For the sake of simplicity let $n=2$. Let 
$\chi_0(k)$ be the
function defined by \eqref{4.3}.  We write $x = x^{(1)}$ and $y = x^{(2)}$. 
We take $\varphi
(x,y)
\in
\Sspace (R^6)$.  Then, we have
\begin{align} \label{4.14}
& \left<\frac{(2\pi)^3}{|V|} \sum_{k\not=0} \frac{\cos k\cdot (x - 
y)}{|k|^2}\chi_0(\epsilon
k),
\varphi (x,y)\right>
\notag \\ & = \frac{(2\pi)^3}{|V|} \sum_{k\not=0} \iint \frac{\cos k\cdot (x -
y)}{|k|^2}\chi_0(\epsilon k) \varphi (x,y)dxdy
\notag \\
&  = \frac{(2\pi)^3}{|V|} \sum_{k\not=0} \iint \frac{\cos k\cdot (x -
y)}{|k|^2<k>^2}\chi_0(\epsilon k) <D_x>^2\varphi (x,y)dxdy,
\end{align}
where $<D_x>^2 = \left(1 - \sum_{j=1}^n \partial_{x_j}^2\right)$.  Let 
$\Phi (k) =
|k|^{-2}<k>^{-2}\displaystyle{\iint} \cos k\cdot (x - y) \\ \times<D_x>^2 
\varphi(x,y)dxdy$
and
\begin{equation*}
   \phi(|k|) := \frac{1}{|k|^2<k>^2} \iint \left| <D_x>^2\varphi (x,y) 
\right|dxdy.
\end{equation*}
Then from \eqref{4.14} Proposition 4.3 shows
\begin{align} \label{4.15}
& \lim_{L_1,L_2,L_3 \rightarrow \infty} \lim_{\epsilon \rightarrow 
0}\left<\frac{(2\pi)^3}{|V|}
\sum_{k\not=0}
\frac{\cos k\cdot (x - y)}{|k|^2}\chi_0(\epsilon k), \varphi (x,y)\right> 
\notag \\
& = \lim_{L_1,L_2,L_3 \rightarrow \infty}\frac{(2\pi)^3}{|V|} 
\sum_{k\not=0} \iint
\frac{\cos k\cdot (x - y)}{|k|^2<k>^2}<D_x>^2\varphi (x,y)dxdy
\notag\\
&  = \int \frac{1}{|k|^2<k>^2} dk \iint \bigl(\cos k\cdot (x - y)\bigr) 
<D_x>^2\varphi
(x,y)dxdy.
\end{align}
In the same way from \eqref{4.14} we also have
\begin{align} \label{4.16}
& \lim_{\epsilon \rightarrow 0} \lim_{L_1,L_2,L_3 \rightarrow \infty}
\left <\frac{(2\pi)^3}{|V|}
\sum_{k\not=0}
\frac{\cos k\cdot (x - y)}{|k|^2}\chi_0(\epsilon k), \varphi (x,y)\right> 
\notag \\
&  = \int \frac{1}{|k|^2<k>^2} dk \iint \bigl(\cos k\cdot (x - y)\bigr) 
<D_x>^2\varphi
(x,y)dxdy.
\end{align}
  On the other hand, Lemma 4.2 and Proposition 4.3 indicate
\begin{align} \label{4.17}
& \lim_{\epsilon \rightarrow 0} \lim_{L_1,L_2,L_3 \rightarrow \infty}
\left<\frac{(2\pi)^3}{|V|}
\sum_{k\not=0}
\frac{\cos k\cdot (x - y)}{|k|^2}\chi_0(\epsilon k), \varphi (x,y)\right> 
\notag \\
& = \lim_{\epsilon \rightarrow 0} \iiint
\frac{\cos k\cdot (x - y)}{|k|^2}\chi_0(\epsilon k) \varphi (x,y) dxdydk 
\notag \\
& = 2\pi^2 \iint \frac{\varphi (x,y)}{|x - y|}dxdy.
\end{align}
Hence we obtain \eqref{4.2} together with \eqref{4.15} and \eqref{4.16}.
\begin{rem}
   Let $\chi (k)  \in \Sspace(R^3)$ such that $\chi(0) = 1$ and $\chi(-k) = 
\chi(k)$.
We take the limit of $L_j\  (j = 1,2,3)$ under the condition \eqref{4.1}. 
Then it holds that
\begin{align} \label{4.18}
& \lim_{\epsilon \rightarrow 0}\lim_{L_1,L_2,L_3 \rightarrow \infty}
\frac{2\pi}{|V|}
\sum_{k \not= 0}\sum_{j,l =1,j \not= l}^n \chi(\epsilon k)\frac{e_j e_l\cos 
k\cdot (x^{(j)} -
x^{(l)})}{|k|^2}
\notag
\\
& = \frac{1}{2} \sum_{j,l=1,j \not= l}^n \frac{e_j e_l}{\bigl| x^{(j)} -
x^{(l)}\bigr|}
\end{align}
pointwise for $\xvec \in R^{3n}$ such that $x^{(j)} - x^{(l)} \not= 0\ (j,l 
= 1,2,\dotsc,n, j
\not= l)$.  The proof is easy.  Consider the case of $n = 2$ and $e_1 = e_2 
= 1$.  Let's
write
$x = x^{(1)}$ and $y = x^{(2)}$.  We take $\chi_1(k) \in C^{\infty}(R^3)$ 
such that $\chi_1(k)
= 1\ (|k| \leq 1)$ and $\chi_1(k) = 0\ (|k| \geq 2)$.  Then, Proposition 
4.3 says for $x \not=
y$ that the lhs of
\eqref{4.18} is equal to
\begin{align} \label{4.19}
& \frac{1}{(2\pi)^2}\lim_{\epsilon \rightarrow 0}
\int \chi(\epsilon k)\frac{\cos k\cdot(x - y)}{|k|^2} dk
\notag
\\
& = \frac{1}{(2\pi)^2}\lim_{\epsilon \rightarrow 0} \Big\{
\int \chi_1( k)\chi(\epsilon k)\frac{\cos k\cdot(x - y)}{|k|^2} dk  \notag \\
& - \frac{1}{|x - y|^2} \int \big(\cos k\cdot(x - y)\big)\Delta_k\left\{ 
\big(1 - \chi_1(
k)\big) \chi(\epsilon k)|k|^{-2}\right\} dk \Big\} \notag \\
& = \frac{1}{(2\pi)^2} \Big\{
\int \chi_1( k)\frac{\cos k\cdot(x - y)}{|k|^2} dk  \notag \\
& - \frac{1}{|x - y|^2} \int \big(\cos k\cdot(x - y)\big)\Delta_k\left\{ 
\big(1 - \chi_1(
k)\big) |k|^{-2}\right\} dk \Big\}
\end{align}
pointwise, where $\Delta_k$ denotes the Laplacian operator with respect to 
$k \in R^3$ and we
used
$|\epsilon \chi'(\epsilon k)| = \epsilon^{1/3}|k|^{-2/3}(\epsilon 
|k|)^{2/3}|\chi'(\epsilon k)|
\leq
\text{Const.} \epsilon^{1/3}|k|^{-2/3}$.  Since we have $\big|\Delta_k \left\{
\big(1 - \chi_1(
k)\big)\chi(\epsilon k) |k|^{-2}\right\}\big| \leq C <k>^{-3-1/3}$ with a 
constant $C$
independent of
$\epsilon$, so we can prove  that the equation \eqref{4.19} is also
true in the distribution sense
$\Sspace'(R^6)$.  On the other hand, we see as in the proof of Lemma 4.2 
that the lhs of
\eqref{4.19} is equal to
$1/(2|x - y|)$ in $\Sspace'(R^6)$. Consequently we can prove that 
\eqref{4.19} is equal to
$1/(2|x - y|)$.  Hence \eqref{4.18} holds pointwise.
\end{rem}
\section{The expression of the vacuum and the states of photons}
In this section we express the vacuum and the states of photons of 
momentums and polarization
states by means of  concrete functions in terms of variables $a_{\Lambda'}$
consisting of the Fourier coefficients of the electromagnetic  potential. 
In Problem 9-8 of
\cite{Feynman-Hibbs} only the vacuum and the state of a photon of momentum 
$\hbar k$ and
polarization state
$\mathrm{l}$ are expressed concretely.  In this section we generalize this 
result  in
\cite{Feynman-Hibbs} for the general states of photons.  In physics the 
vacuum and the state of
photons are not considered concretely but considered abstractly  (cf. \cite{Sakurai,Swanson}). We also 
note that the state
of photons of given momentums and polarization states   can't be considered 
in the study for
QED models on  the Fock spaces (cf.
\cite{Gustafson-Sigal,Hiroshima,Spohn}).
\par
To write down the vacuum and the state of photons concretely, we will 
introduce the  creation
operators and the annihilation operators concretely.
Let's define
\begin{align} \label{5.1}
\hat{a}_{\mathrm{l}k}^{(\mathrm{i})} &:= i\sqrt{\frac{|V|}{2\hbar c|k|}}
\left(\frac{\hbar}{i}
\frac{\partial}{\partial
\alki} - i\frac{c|k|}{|V|}\alki\right) \notag \\
& = \sqrt{\frac{|V|}{2\hbar c|k|}}
\left(\hbar
\frac{\partial}{\partial
\alki} + \frac{c|k|}{|V|}\alki\right)
\end{align}
for $k \in \Lambda$ and $\mathrm{i,l} = 1,2$.  From \eqref{2.13} we have
\begin{equation*}
\hat{a}_{\mathrm{l}-k}^{(\mathrm{1})} = 
-\hat{a}_{\mathrm{l}k}^{(\mathrm{1})}, \
\hat{a}_{\mathrm{l}-k}^{(\mathrm{2})} = \hat{a}_{\mathrm{l}k}^{(\mathrm{2})}.
\end{equation*}
Let $\hat{a}_{\mathrm{l}k}^{(\mathrm{i})\dag}$ denote the adjoint operator of
$\hat{a}_{\mathrm{l}k}^{(\mathrm{i})}$. Then we know that the commutator 
relations
\begin{equation*}
[\hat{a}_{\mathrm{l}k}^{(\mathrm{i})}, 
\hat{a}_{\mathrm{l'}k'}^{(\mathrm{i'})\dag}] =
\delta_{\mathrm{i'}\mathrm{i}}\delta_{\mathrm{l'}\mathrm{l}}\delta_{kk'},
\quad  [\hat{a}_{\mathrm{l}k}^{(\mathrm{i})}, 
\hat{a}_{\mathrm{l'}k'}^{(\mathrm{i'})}] = 0
\end{equation*}
hold for $k$ and $k'$ in $\Lambda'$ (cf. \S 34 in \cite{Dirac}).
We define the operator $\hat{a}_{\mathrm{l}k}$ for $k \in \Lambda$ and 
$\mathrm{l} = 1,2$  by
\begin{equation} \label{5.2}
\hat{a}_{\mathrm{l}k} := \frac{\hat{a}_{\mathrm{l}k}^{(\mathrm{1})} -
i\hat{a}_{\mathrm{l}k}^{(\mathrm{2})}}{\sqrt{2}}
\end{equation}
(cf. \eqref{2.11}).
We call $\hat{a}_{\mathrm{l}k}$  the annihilation operator and
$\hat{a}_{\mathrm{l}k}^{\dag}$ the creation operator. We can easily see 
from the
commutator relations for $\hat{a}_{\mathrm{l}k}^{(\mathrm{i})}$ that the 
operators
$\hat{a}_{\mathrm{l}k}$ and
$\hat{a}_{\mathrm{l}k}^{\dag}$ also satisfy the commutator relations
\begin{equation} \label{5.3}
[\hat{a}_{\mathrm{l}k}, \hat{a}_{\mathrm{l'}k'}^{\dag}] =
\delta_{\mathrm{l'}\mathrm{l}}\delta_{kk'},
\quad  [\hat{a}_{\mathrm{l}k}, \hat{a}_{\mathrm{l'}k'}] = 0
\end{equation}
for $k$ and $k'$ in $\Lambda$ (cf. (2.26) in
\cite{Sakurai}).  It follows from the commutator relations \eqref{5.3} that 
we have
\begin{equation} \label{5.4}
\hat{a}_{\mathrm{l}k} (\hat{a}_{\mathrm{l}k}^{\dag})^{n'}
- (\hat{a}_{\mathrm{l}k}^{\dag})^{n'} \hat{a}_{\mathrm{l}k}
= n' (\hat{a}_{\mathrm{l}k}^{\dag})^{n'-1}
\end{equation}
  (cf. \S 34 in \cite{Dirac}). Then we get the following expression as in 
physics (cf. p.198 in
\cite{Gustafson-Sigal}, (2.60) and (2.64) in \cite{Sakurai}).
\begin{pro}
   We can write the last term of $H(t)$ defined by \eqref{3.10} as
\begin{align} \label{5.5}
H_{\text{rad}} &:=  \sum_{k \in \Lambda',\mathrm{l}}
\sum_{\mathrm{i}=1}^2\left\{
\frac{|V|}{2} \left(\frac{\hbar}{i} \frac{\partial}{\partial
\alki}\right)^2 +\frac{(c|k|)^2}{2|V|} \left(\alki\right)^2 -
\frac{\hbar c|k|}{2}
\right\} \notag \\
& = \sum_{k \in \Lambda,\mathrm{l}} \hbar c|k|\hat{a}_{\mathrm{l}k}^{\dag}
\hat{a}_{\mathrm{l}k}.
\end{align}
The vector potential $A(x,a_{\Lambda'_2})$ defined by \eqref{2.9} or \eqref{2.14}, where
the sum of $k$ is taken over $\Lambda_2$, is given by the expression
\begin{equation} \label{5.6}
A(x,a_{\Lambda'_2}) = \sqrt{\frac{4\pi\hbar}{|V|}}c \sum_{k \in
\Lambda_2}\sum_{\mathrm{l}=1}^2 \frac{1}{\sqrt{2c|k|}}
\bigl( \hat{a}_{\mathrm{l}k} \mathrm{e}^{ik\cdot x}  +
   \hat{a}_{\mathrm{l}k}^{\dag} \mathrm{e}^{-ik\cdot x} \bigr)
\overrightarrow{e_{\mathrm{l}}}(k).
\end{equation}
\end{pro}
\begin{proof}
Since from \eqref{5.1} and \eqref{5.2} we have
\begin{align*}
& \hbar c|k|(\hat{a}_{\mathrm{l}k}^{\dag} \hat{a}_{\mathrm{l}k} +
\hat{a}_{\mathrm{l}-k}^{\dag} \hat{a}_{\mathrm{l}-k}) \\
& = \frac{\hbar c|k|}{2}\left\{\left(\hat{a}_{\mathrm{l}k}^{(1)\dag} +
i\hat{a}_{\mathrm{l}k}^{(2)\dag}\right)
\left(\hat{a}_{\mathrm{l}k}^{(1)} - i\hat{a}_{\mathrm{l}k}^{(2)}\right)
+ \left(-\hat{a}_{\mathrm{l}k}^{(1)\dag} +
i\hat{a}_{\mathrm{l}k}^{(2)\dag}\right)
\left(-\hat{a}_{\mathrm{l}k}^{(1)} - 
i\hat{a}_{\mathrm{l}k}^{(2)}\right)\right\}\\
& = \hbar 
c|k|\left(\hat{a}_{\mathrm{l}k}^{(1)\dag}\hat{a}_{\mathrm{l}k}^{(1)} +
\hat{a}_{\mathrm{l}k}^{(2)\dag}\hat{a}_{\mathrm{l}k}^{(2)}\right) \\
& = \sum_{\mathrm{i}=1}^2\left\{
\frac{|V|}{2} \left(\frac{\hbar}{i} \frac{\partial}{\partial
\alki}\right)^2 +\frac{(c|k|)^2}{2|V|} (\alki)^2 -
\frac{\hbar c|k|}{2}
\right\}
\end{align*}
for $k \in \Lambda$, so we get \eqref{5.5}. \par
   From \eqref{5.1} and \eqref{5.2} we have
\begin{align*}
& \hat{a}_{\mathrm{l}k}\mathrm{e}^{ik\cdot x} +
\hat{a}_{\mathrm{l}k}^{\dag}\mathrm{e}^{-ik\cdot x} \\
& = \frac{1}{\sqrt{2}}\Bigl\{\left(\hat{a}_{\mathrm{l}k}^{(1)} +
\hat{a}_{\mathrm{l}k}^{(1)\dag}\right)\cos k\cdot x -
i \left(\hat{a}_{\mathrm{l}k}^{(2)} - 
\hat{a}_{\mathrm{l}k}^{(2)\dag}\right)\cos k\cdot x \\
& + i\left(\hat{a}_{\mathrm{l}k}^{(1)} -
\hat{a}_{\mathrm{l}k}^{(1)\dag}\right)\sin k\cdot x
+ \left(\hat{a}_{\mathrm{l}k}^{(2)} + 
\hat{a}_{\mathrm{l}k}^{(2)\dag}\right)\sin k\cdot
x \Bigr\} \\
& = \sqrt{\frac{|V|}{\hbar 
c|k|}}\Biggl(\frac{c|k|}{|V|}a_{\mathrm{l}k}^{(1)}\cos k\cdot x
+ \frac{c|k|}{|V|}a_{\mathrm{l}k}^{(2)}\sin k\cdot x \\
& -i\hbar (\cos k\cdot x) \frac{\partial}{\partial a_{\mathrm{l}k}^{(2)}}
+  i\hbar (\sin k\cdot x) \frac{\partial}{\partial 
a_{\mathrm{l}k}^{(1)}}\Biggr).
\end{align*}
So, it is shown from \eqref{2.8} and \eqref{2.13} that
\begin{align*}
& \sum_{k \in \Lambda_2}\frac{1}{\sqrt{2c|k|}} 
\left(\hat{a}_{\mathrm{l}k}\mathrm{e}^{ik\cdot
x} + \hat{a}_{\mathrm{l}k}^{\dag}\mathrm{e}^{-ik\cdot x}\right)
\overrightarrow{e_{\mathrm{l}}}(k) \\
& = \sum_{k \in \Lambda_2}\frac{1}{\sqrt{2\hbar|V|}} 
\left(a_{\mathrm{l}k}^{(1)}\cos k\cdot x +
a_{\mathrm{l}k}^{(2)}\sin k\cdot x\right)
\overrightarrow{e_{\mathrm{l}}}(k).
\end{align*}
Hence, we see that the rhs of \eqref{5.6} is equal to
\begin{equation*}
\frac{\sqrt{4\pi}}{|V|}c\sum_{k \in 
\Lambda_2}\sum_{\mathrm{l}=1}^2\frac{1}{\sqrt{2}}
\left(a_{\mathrm{l}k}^{(1)}\cos k\cdot x + a_{\mathrm{l}k}^{(2)}\sin k\cdot 
x\right)
\overrightarrow{e_{\mathrm{l}}}(k),
\end{equation*}
which is equal to the lhs of \eqref{5.6} from \eqref{2.14}.
\end{proof}
\par
We know
\[
\int_{-\infty}^{\infty} \mathrm{e}^{-a\theta^2}d\theta = \sqrt{\frac{\pi}{a}}
\]
for a constant $a > 0$.  So, we can easily see from \eqref{5.2} and 
\eqref{5.5} that
\begin{equation} \label{5.7}
\Psi_0 (a_{\Lambda'}):= \prod_{k \in 
\Lambda',\mathrm{l}}\sqrt{\frac{c|k|}{\pi\hbar |V|}}
\exp
\left\{ - \frac{c|k|}{2\hbar |V|}\left(a_{\mathrm{l}k}^{(\mathrm{1})^2} +
a_{\mathrm{l}k}^{(\mathrm{2})^2}\right) \right\}
\end{equation}
is the normal ground state of $H_{\text{rad}}$, called vacuum, whose energy is
$0$, i.e.
\begin{equation} \label{5.8}
H_{\text{rad}}\Psi_0 = 0
\end{equation}
and that we have
\begin{equation} \label{5.9}
\hat{a}_{\mathrm{l}k}^{\dag} \Psi_0 = \sqrt{\frac{2c|k|}{\hbar |V|}}
a_{\mathrm{l}k}^*
\Psi_0, \  \hat{a}_{\mathrm{l}k} \Psi_0 = 0\quad (k \in \Lambda)
\end{equation}
(cf. \S 8-1, (9-43)  and Problem 9-8 in \cite{Feynman-Hibbs}).  We know that
the eigenvalue $0$ of \eqref{5.8} is simple  (cf.
Theorem 3.4 in Chapter 3 of \cite{Berezin-Shubin}).
\par
The function
$\Psi_{n'\mathrm{l}k} (a_{\Lambda'}):= (\hat{a}_{\mathrm{l}k}^{\dag})^{n'}
\Psi_0(a_{\Lambda'})\
(k \in \Lambda, n' = 0,1,2,\dotsc)$, which can be written  concretely from
\eqref{5.1}, \eqref{5.2} and \eqref{5.7}, expresses the state of
$n'$ photons of  momentum $\hbar k$ and  polarization state
$\mathrm{l}$ (cf. \S 9-2 in \cite{Feynman-Hibbs} and \S 2-2 in 
\cite{Sakurai}) and
satisfies
\begin{equation*}
\left(\sum_{k \in \Lambda, \mathrm{l}} \hat{a}_{\mathrm{l}k}^{\dag}
\hat{a}_{\mathrm{l}k}\right)
\Psi_{n'\mathrm{l'}k'} = n' \Psi_{n'\mathrm{l'}k'},
\end{equation*}
\begin{equation*}
\left(\sum_{k \in \Lambda} \hbar
k\hat{a}_{\mathrm{l}k}^{\dag}\hat{a}_{\mathrm{l}k} \right)
\Psi_{n'\mathrm{l'}k'} = n'(\hbar k') \Psi_{n'\mathrm{l'}k'}
\end{equation*}
and
\begin{equation*}
H_{\text{rad}} \Psi_{n'\mathrm{l'}k'}  =
n' (\hbar c|k'|)\Psi_{n'\mathrm{l'}k'}
\end{equation*}
from \eqref{5.4}, \eqref{5.5} and \eqref{5.9}.  The operators $\sum_{k \in 
\Lambda, \mathrm{l}}
\hat{a}_{\mathrm{l}k}^{\dag}
\hat{a}_{\mathrm{l}k}$ and $\sum_{k \in \Lambda} \hbar
k\hat{a}_{\mathrm{l}k}^{\dag}\hat{a}_{\mathrm{l}k}$ are called the total
number operator and the momentum operator, respectively (cf. (2.68) and 
(2.80) in
\cite{Sakurai}).
In the same way, $\prod_{k \in \Lambda,\mathrm{l}}
(\hat{a}_{\mathrm{l}k}^{\dag})^{n'(\mathrm{l},k)}\Psi_0(a_{\Lambda'})\ 
(n'(\mathrm{l},k) =
0,1,\dotsc)$ denotes the state of
$n'(\mathrm{l},k)$ photons of momentum $\hbar k$ and polarization state 
$\mathrm{l}$.
Then, setting $\Psi(a_{\Lambda'}) = \prod_{k \in \Lambda,\mathrm{l}}
(\hat{a}_{\mathrm{l}k}^{\dag})^{n'(\mathrm{l},k)}\Psi_0(a_{\Lambda'})$, we get
\begin{equation} \label{5.10}
\left(\sum_{k \in \Lambda, \mathrm{l}} \hat{a}_{\mathrm{l}k}^{\dag}
\hat{a}_{\mathrm{l}k}\right)
\Psi = \left(\sum_{k \in \Lambda, \mathrm{l}} n'(\mathrm{l},k) \right)\Psi,
\end{equation}
\begin{equation} \label{5.11}
\left(\sum_{k \in \Lambda} \hbar
k\hat{a}_{\mathrm{l}k}^{\dag}\hat{a}_{\mathrm{l}k} \right)
\Psi = \left(\sum_{k \in \Lambda, \mathrm{l}} n'(\mathrm{l},k)\hbar k 
\right)\Psi
\end{equation}
and
\begin{equation} \label{5.12}
H_{\text{rad}} \Psi =
\left(\sum_{k \in \Lambda, \mathrm{l}} n'(\mathrm{l},k)\hbar c |k| \right)\Psi.
\end{equation}
\par
   The family
\[
\left\{ \prod_{k \in \Lambda', \mathrm{l},\mathrm{i}}
\left(\hat{a}_{\mathrm{l}k}^{\mathrm{(i)}\dag}\right)^{n'(\mathrm{l},k,
\mathrm{i})}
\Psi_0
\right\}_{n'(\mathrm{l},k,\mathrm{i}) = 0}^{\infty}
\]
makes a complete orthogonal system in $L^2(R^{4N})$ (cf. Theorem
3.1 in Chapter 3 of \cite{Berezin-Shubin} and \S 34 in \cite{Dirac}).  We have
\begin{equation*}
\hat{a}_{\mathrm{l}k}^{(1)} = \frac{\hat{a}_{\mathrm{l}k} -
\hat{a}_{\mathrm{l}-k}}{\sqrt{2}}, \quad
\hat{a}_{\mathrm{l}k}^{(2)} = \frac{i(\hat{a}_{\mathrm{l}k} +
\hat{a}_{\mathrm{l}-k})}{\sqrt{2}}
\end{equation*}
from \eqref{2.13} and \eqref{5.2}.    So we see together with \eqref{5.4} and
the second equation in \eqref{5.9} that
the family
\begin{equation} \label{5.13}
\left\{ \prod_{k \in \Lambda, \mathrm{l}} \frac{1}{\sqrt{n'(\mathrm{l},k)!}}
\left(\hat{a}_{\mathrm{l}k}^{\dag}\right)^{n'(\mathrm{l},k)} \Psi_0
\right\}_{n'(\mathrm{l},k) = 0}^{\infty}
\end{equation}
also makes a complete orthonormal system in $L^2(R^{4N})$ (cf. \S 34 in 
\cite{Dirac} and
(2.46) in \cite{Sakurai}). For example, we have
\begin{align*}
\left(\hat{a}_{\mathrm{l}k}^{\dag}\Psi_0,
\left(\hat{a}_{\mathrm{l}k}^{\dag}\right)^{2}\Psi_0\right) & =   \left(\Psi_0,
\hat{a}_{\mathrm{l}k}\left(\hat{a}_{\mathrm{l}k}^{\dag}\right)^{2}\Psi_0\right)    \\
& = \left(\Psi_0, \left(\hat{a}_{\mathrm{l}k}^{\dag}\right)^{2}
\hat{a}_{\mathrm{l}k}\Psi_0\right) + 2\left(\Psi_0, 
\hat{a}_{\mathrm{l}k}^{\dag}\Psi_0\right)
\\
& = 2\left(\hat{a}_{\mathrm{l}k}\Psi_0, \Psi_0\right) = 0.
\end{align*}
\begin{rem}
We considered the Lagrangian function  \eqref{3.3} and the Hamiltonian 
operator \eqref{3.10},
determining an indefinite constant in \eqref{2.3} by \eqref{2.18} or in
Remark 3.1.  On the other hand, in many literatures (cf.
\cite{Feynman-Hibbs},
\cite{Sakurai} and
\cite{Spohn})
   an indefinite constant is determined
to be 0.
Consequently, the  term $\infty = (1/2) \sum_{j =1}^n e_j^2/| x^{(j)} - 
x^{(j)}|$
appears in \eqref{4.2} from \eqref{2.21}
  and
the ground state energy of $H_{\text{rad}}$ is $\sum_{k \in \Lambda'} \hbar
c|k|/2$, which
tends to infinity as
$M_3$ tends to infinity.  Arguments are had about these infinities in \S 
9-3 and \S 9-5 of
\cite{Feynman-Hibbs}.  In the present paper we could see that the term
$(1/2)
\sum_{j =1}^n e_j^2/| x^{(j)} - x^{(j)}|$ disappears in \eqref{4.2} and 
that the ground
state energy  of
$H_{\text{rad}}$ is $0$.
\end{rem}
   \section{Preliminaries for the proofs of main results}
 From \S 6 to \S 9 we often write $\xvec$ and $\yvec$ in $R^{3n}$ as $x$ and 
$y$, respectively
for the sake of simplicity when no confusion arises. \par
   Let $0 \leq s \leq t \leq T$.  For $x$ and $y$ in $R^{3n}$ we define
\begin{equation} \label{6.1}
      \qts(\theta) := x - \frac{t - \theta}{t - s}(x - y), \   s\leq \theta 
\leq t.
\end{equation}
  For $X$ and $Y$ in $R^{4N}$ we also define
\begin{equation} \label{6.2}
      a_{\Lambda' X,Y}^{t,s}(\theta) := X - \frac{t - \theta}{t - s}(X - Y), 
\   s\leq \theta
\leq t.
\end{equation}
Then  $a_{\Lambda X,Y}^{t,s}(\theta) \in R^{8N}$ is defined by means of 
\eqref{2.13}.  We set
\begin{equation} \label{6.3}
      V_1(x)  := \frac{2\pi}{|V|}
\sum_{k \in \Lambda_1}\sum_{j,l =1,j \not= l}^n\frac{e_j e_l\cos k\cdot 
(x^{(j)} -
x^{(l)})}{|k|^2}
\end{equation}
and
\begin{equation} \label{6.4}
      V_2(a_{\Lambda'})  := \sum_{k \in \Lambda',\mathrm{i,l}} \left(
\frac{(c|k|)^2}{2|V|} \left(\alki \right)^2 -
\frac{\hbar c|k|}{2}
\right).
\end{equation} \par
For the sake of simplicity we suppose $\Lambda_2' = \Lambda_3'( = 
\Lambda')$ from \S 6 to \S
9.  We write $\bold x = (x, X) \in R^{3n + 4N}$ and
\begin{equation}  \label{6.5}
   {\bold q}^{t,s}_{{\bold x},{\bold y}}(\theta) = (\theta,
\qts(\theta), a_{\Lambda' X,Y}^{t,s}(\theta))
\in R^{1 + 3n + 4N}, \  s \leq \theta \leq t.
\end{equation}
Then from \eqref{3.3} and \eqref{3.5} we have
\begin{align}  \label{6.6}
  & S_c(t,s;\qts,a_{\Lambda X,Y}^{t,s})  = \frac{1}{2(t - s)}\sum_{j=1}^n 
m_j|x^{(j)} -
y^{(j)}|^2  \notag \\
& \quad +
\int_{{\bold q}^{t,s}_{{\bold x},{\bold y}}}\left( - V_1(x)dt + \frac{1}{c}
\sum_{j=1}^n e_j\tilde{A}(x^{(j)},a_{\Lambda'})\cdot dx^{(j)} - 
V_2(a_{\Lambda'})dt \right) +
\frac{|X- Y|^2}{2|V|(t - s)}\notag \\
& = \frac{1}{2(t - s)}\sum_{j=1}^n m_j|x^{(j)} - y^{(j)}|^2 -  \int^t_s V_1(x
- \frac{t - \theta}{t - s} (x - y))d\theta
   \notag \\
   & \quad + \frac{1}{c} \sum_{j=1}^n e_j(x^{(j)} - y^{(j)})\cdot \int^1_0 
\tilde{A}(x^{(j)}
- \theta (x^{(j)} - y^{(j)}),X- \theta (X - Y))d\theta   \notag
\\
   & \qquad + \frac{|X- Y|^2}{2|V|(t - s)} - \int^t_s V_2(X- \frac{t - 
\theta}{t - s} (X -
Y))d\theta
\notag \\
       & = \frac{1}{2(t - s)}\sum_{j=1}^n m_j|x^{(j)} - y^{(j)}|^2 - (t - s) 
\int^1_0  V_1(x
- \theta (x - y))d\theta
   \notag \\
   & \quad + \frac{1}{c} \sum_{j=1}^n e_j(x^{(j)} - y^{(j)})\cdot \int^1_0 
\tilde{A}(x^{(j)}
- \theta (x^{(j)} - y^{(j)}),X- \theta (X - Y))d\theta   \notag
\\
   & \qquad + \frac{|X- Y|^2}{2|V|(t - s)} - (t - s)\int^1_0 V_2(X- \theta 
(X - Y))d\theta.
\end{align} \par
   Let $M \geq 0$ and $p(x,w,X,W)$ a $\Cspace$ function in $R^{6n}\times
R^{8N}$ such that
\begin{equation}  \label{6.7}
   |\partial^{\alpha}_w\partial^{\beta}_x 
\partial^{\alpha'}_W\partial^{\beta'}_X p(x,w,X,W)|
\leq C_{\alpha,\beta,\alpha',\beta'}\left(<x;w> <X;W>\right)^M
\end{equation}
for all multi-indices $\alpha, \beta, \alpha'$ and $\beta'$ with constants
$C_{\alpha,\beta,\alpha',\beta'}$,  where $<x;w> := \sqrt{1 + |x|^2 + 
|w|^2}$.  For $f(x,X)
\in \Sspace (R^{3n+4N})$  we define the operator $P(t,s)$ by
\begin{equation}  \label{6.8}
        \begin{cases}
            \begin{split}
              & \left(\prod_{j=1}^n\sqrt{\frac{m_j}{2\pi i\hbar(t - s)}}^{\ 
3}\right)
               \sqrt{\frac{1}{2\pi i\hbar|V|(t - s)}}^{\ 4N}  \iint \left( \exp
i\hbar^{-1}S_c(t,s;\qts,a_{\Lambda X,Y}^{t,s}\right)
\\ & \qquad\  \times p\left(x,\frac{x-y}{\sqrt{t - s}},X,\frac{X-Y}{\sqrt{t -
s}}\right) f(y,Y)dydY,
                      \end{split}
         & s < t ,
                  \\
\begin{split}
        & \left(\prod_{j=1}^n \sqrt{\frac{m_j}{2\pi i\hbar}}^{\ 3}\right)
\sqrt{\frac{1}{2\pi i\hbar|V|}}^{\ 4N}
\text{Os}-\iint \Bigl\{\exp i\hbar^{-1}\Bigl(\sum_{j=1}^n \frac{m_j|w_j|^2}{2}
   \\
& \qquad \   +  \frac{|W|^2}{2|V|}\Bigr) \Bigr\}p(x,w,X,W)dwdW f(x,X),\\
             \end{split}
            & s = t.
        \end{cases}
\end{equation}
When $p(x,w,X,W) = 1$, $P(t,s)$ is called {\it the fundamental operator} 
and denoted by $\Cts$.
\begin{lem}
Let $M_1$ and $M_2$ be non-negative constants.  Suppose that $g(x) \ (x \in 
R^3)$ and $\psi
(\theta)\ (\theta \in R)$ in  \eqref{3.4} satisfy
\begin{equation*}
   |\partial^{\alpha}_x g(x)|
\leq C_{\alpha}<x>^{M_1},\ x \in  R^3
\end{equation*}
for all $\alpha$ and
\begin{equation*}
   |\frac{d^k}{d\theta^k} \psi (\theta)|
\leq C_{k}<\theta>^{M_2},\ \theta \in  R
\end{equation*}
for all $k = 0,1,\dotsc$.  Let $f \in \Sspace(R^{3n+4N})$.  Then,
$\partial^{\alpha}_x\partial^{\alpha'}_X(P(t,s)f)(x,X)$ are continuous in 
$0 \leq s
\leq t \leq T$ and $(x,X) \in R^{3n+4N}$ for all  $\alpha$ and $\alpha'$.
\end{lem}
\begin{proof}
Let $s < t$ and make the change of variables: $y \rightarrow w = (x - 
y)/\sqrt{t-s}$ and $Y
\rightarrow W = (X - Y)/\sqrt{t-s}$ in \eqref{6.8}.  Then from \eqref{6.6} 
we have
\begin{align}  \label{6.9}
       & P(t,s)f = \left(\prod_{j=1}^n \sqrt{\frac{m_j}{2\pi i\hbar}}^{\ 
3}\right)
\sqrt{\frac{1}{2\pi i\hbar|V|}}^{\ 4N}
     \text{Os}-\iint \left(\exp i\hbar^{-1}\phi(t,s;x,w,X,W) \right)\notag\\
& \times p(x,w,X,W)f(x - \sqrt{\rho}w,X - \sqrt{\rho}W)dwdW,\quad \rho = t - s,
\end{align}
where
\begin{align}  \label{6.10}
  &\phi(t,s;x,w,X,W)  := \sum_{j=1}^n\frac{m_j}{2}|w^{(j)}|^2 + 
\frac{1}{2|V|}|W|^2 +
\psi(t,s;x,\sqrt{\rho}w,X,\sqrt{\rho}W)
     \notag    \\
             & := \sum_{j=1}^n\frac{m_j}{2}|w^{(j)}|^2 + \frac{1}{2|V|}|W|^2 
- \rho\int^1_0
V_1( x-\theta\sqrt{\rho}w)d\theta + \frac{1}{c}
\sum_{j=1}^n e_j
\sqrt{\rho}w^{(j)} \notag \\
& \cdot
\int^1_0 \tilde{A}(x^{(j)} - \theta\sqrt{\rho}w^{(j)}, 
X-\theta\sqrt{\rho}W)d\theta -
\rho\int^1_0
V_2( X-\theta\sqrt{\rho}W)d\theta.
\end{align}
We note from \eqref{6.8} that \eqref{6.9} is also true for $t = s$. \par
Let $L^{(j)} := <w^{(j)}>^{-2}\left(1 - i\hbar m_j^{-1} \sum_{k=1}^3
w^{(j)}_k\partial/\partial w^{(j)}_k\right)\ (j = 1,2,\dotsc,n)$ and
${}^tL^{(j)}$ its transposed operator.  We also let $L_1 := <W>^{-2}\left(1 
- i\hbar
|V| \sum_{k=1}^{4N} W_k\partial_{W_k}\right)$.  Then, integrating by parts 
with respect to
$w$ and $W$ in \eqref{6.9} by means of $L^{(j)}$ and $L_1$, we see that the 
integrand is
bounded by
\begin{equation*}
\text{Const.} <x;X>^l <w>^{-(3n+1)}<W>^{-(4N+1)}
\end{equation*}
for some real constant $l$.  See the proof of Lemma 2.1 in \cite{Ichinose 
1999} for further
details.  Consequently, we see that $\left(P(t,s)f\right)(x,X)$ is 
continuous  in $0 \leq s
\leq t
\leq T$ and $(x,X) \in R^{3n+4N}$. In the same way we can complete the proof.
\end{proof} \par
    For $0 \leq \sigma_1, \sigma_2 \leq 1$  we set $\sigma := 
(\sigma_1,\sigma_2) $ and
\begin{align}   \label{6.11}
     & \tau(\sigma) :=  t - \sigma_1(t - s) \in R,        \notag \\
       &  \zeta^{(j)}(\sigma) := z^{(j)} + \sigma_1(x^{(j)} - z^{(j)})
                  + \sigma_1\sigma_2(y^{(j)} - x^{(j)}) \in R^3, \ j = 
1,2,\dotsc,n, \notag\\
         &  \zeta(\sigma) := z + \sigma_1(x - z)
                  + \sigma_1\sigma_2(y - x) \in R^{3n},  \notag\\
&  \tilde{\zeta}(\sigma) := Z + \sigma_1(X - Z)
                  + \sigma_1\sigma_2(Y - X) \in R^{4N}.
\end{align}
We also set
\begin{equation} \label{6.12}
B_{ml}(x^{(j)},a_{\Lambda'}) = \frac{\partial \tilde{A}_l}{\partial
x_m}(x^{(j)},a_{\Lambda'}) - \frac{\partial \tilde{A}_m}{\partial
x_l}(x^{(j)},a_{\Lambda'})
\end{equation}
for $l,m = 1,2,3$ and $j = 1,2,\dotsc,n$.  Then from \eqref{6.6} we have
\begin{lem}
We  can write
\begin{align}   \label{6.13}
     & S_c(t,s;\overrightarrow{q}^{t,s}_{z,y},a_{\Lambda Z,Y}^{t,s}) -
S_c(t,s;\overrightarrow{q}^{t,s}_{z,x},a_{\Lambda Z,X}^{t,s})   \notag \\
       &  = \frac{1}{t - s}\sum_{j=1}^n m_j(x^{(j)} - 
y^{(j)})\cdot\left(z^{(j)} -
\frac{x^{(j)} + y^{(j)}}{2}\right)
  \notag\\
& + (t - s) (x - y)\cdot \int_0^1\int_0^1 \sigma_1 \frac{\partial V_1}{\partial
x}\left(\zeta(\sigma)\right) d\sigma_1d\sigma_2 \notag \\
&  + \frac{1}{c} \sum_{j=1}^n
e_j(x^{(j)} - y^{(j)})\cdot \int^1_0 \tilde{A}(x^{(j)} - \theta (x^{(j)} - 
y^{(j)}),X- \theta
(X - Y))d\theta \notag \\
&  + \frac{1}{c} \sum_{j=1}^n \sum_{l,m=1}^3
e_j(x^{(j)}_m - y^{(j)}_m)(x^{(j)}_l - z^{(j)}_l) \int^1_0\int^1_0 \sigma_1 
B_{ml}
(\zeta^{(j)}(\sigma),\tilde{\zeta}(\sigma))d\sigma_1 d\sigma_2
\notag \\
& + \frac{1}{c} \sum_{j=1}^n
e_j(x^{(j)} - y^{(j)})\cdot \left\{ (Z - X)\cdot\int^1_0\int^1_0 \sigma_1
\frac{\partial\tilde{A}}{\partial
a_{\Lambda'}}(\zeta^{(j)}(\sigma),\tilde{\zeta}(\sigma))d\sigma_1 d\sigma_2 
\right\}
\notag \\
& + (X - Y)\cdot \frac{1}{c} \sum_{j=1}^n \sum_{m=1}^3
e_j(x^{(j)}_m - z^{(j)}_m)  \int^1_0\int^1_0 \sigma_1
\frac{\partial\tilde{A}_m}{\partial
a_{\Lambda'}}(\zeta^{(j)}(\sigma),\tilde{\zeta}(\sigma))d\sigma_1 d\sigma_2
\notag \\
&  + \frac{1}{(t - s)|V|}(X - Y)\cdot \left(Z -
\frac{X + Y}{2}\right)    \notag\\
& + (t-s)(X - Y)\cdot  \int^1_0\int^1_0 \sigma_1
\frac{\partial V_2}{\partial
a_{\Lambda'}}(\tilde{\zeta}(\sigma))d\sigma_1 d\sigma_2.
\end{align}
\end{lem}
\begin{proof}
   We use \eqref{6.6}. From \eqref{6.5} and \eqref{6.11} we see
\begin{align}  \label{6.14}
& \int_{{\pmb q}_{{\pmb z},{\pmb y}}^{t,s}} (-V_1(x))dt - \int_{{\pmb 
q}_{{\pmb z},
{\pmb x}}^{t,s}} (-V_1(x))dt = \sum_{j=1}^n \sum_{l=1}^3
\iint_{\bDelta}\partial V_1/\partial x^{(j)}_l dt \wedge dx^{(j)}_l
  \notag \\
& =  \sum_{j=1}^n \sum_{l=1}^3 \int_0^1\int_0^{1} \partial V_1 
(\zeta(\sigma))/\partial
x^{(j)}_l
\det\frac{\partial(\tau(\sigma),\zeta^{(j)}_l(\sigma))}{\partial(\sigma_1,
\sigma_2)}
d\sigma_1d\sigma_2
\notag\\
& =  \sum_{j=1}^n \sum_{l=1}^3 (t - s)\left(x^{(j)}_l - y^{(j)}_l\right) 
\int_0^1\int_0^{1}
\sigma_1 \partial V_1 (\zeta(\sigma))/\partial x^{(j)}_l
d\sigma_1d\sigma_2 \notag \\
       & =   (t - s)(x - y)\cdot  \int_0^1\int_0^{1}
            \sigma_1 \frac{\partial V_1}{\partial x} (\zeta(\sigma))
d\sigma_1d\sigma_2,
\end{align}
where $\bDelta = \bDelta (t,s,x,y,z)$ is the 2-dimensional plane with
oriented boundary
consisting of $(\theta,\qvec^{t,s}_{z,y}(\theta)), 
-(\theta,\qvec^{t,s}_{z,x}(\theta))$ and
$(\theta,\qvec^{s,s}_{y,x}(\theta))\ (s \leq \theta \leq t)$, and $\sigma$ 
in \eqref{6.11}
gives the positive orientation of $\Delta$.  So the second term on the rhs 
of \eqref{6.13}
appears.  In the same way the last term appears.  It is easy to show that 
the first  and the
7th terms appear.
\par
   As in the proof of \eqref{6.14} we have
\begin{align}  \label{6.15}
& \int_{{\pmb q}_{{\pmb z},{\pmb 
y}}^{t,s}}\tilde{A}(x^{(j)},a_{\Lambda'})\cdot dx^{(j)}
- \int_{{\pmb q}_{{\pmb z},{\pmb 
x}}^{t,s}}\tilde{A}(x^{(j)},a_{\Lambda'})\cdot dx^{(j)}
\notag\\
& = \int_{{\pmb q}_{{\pmb x},{\pmb 
y}}^{s,s}}\tilde{A}(x^{(j)},a_{\Lambda'})\cdot dx^{(j)}
+ \iint_{\bDelta}d\left(\tilde{A}(x^{(j)},a_{\Lambda'})\cdot dx^{(j)}\right)
\notag\\
& = (x^{(j)} - y^{(j)}) \cdot \int_0^1 \tilde{A}\left(x^{(j)} - 
\theta(x^{(j)} -
y^{(j)}),X - \theta(X- Y)\right) d\theta \notag \\
& + \sum_{1 \leq m < l \leq 3} \iint_{\bDelta}B_{ml} dx^{(j)}_m \wedge 
dx^{(j)}_l
- \sum_{k\in \Lambda',\mathrm{i},\mathrm{l}} \sum_{m=1}^3
\iint_{\bDelta}\left(\partial \tilde{A}_m/\partial \alki\right) dx^{(j)}_m 
\wedge d\alki
\notag \\
& = (x^{(j)} - y^{(j)}) \cdot \int_0^1 \tilde{A}\left(x^{(j)} - 
\theta(x^{(j)} -
y^{(j)}),X - \theta(X- Y)\right) d\theta
\notag\\
& + \sum_{1 \leq m < l \leq 3}\left\{(x^{(j)}_m - y^{(j)}_m)(x^{(j)}_l - 
z^{(j)}_l)
- (x^{(j)}_l - y^{(j)}_l)(x^{(j)}_m - z^{(j)}_m) \right\} \notag \\
& \quad \times
\int_0^1\int_0^1 
\sigma_1B_{ml}\left(\zeta^{(j)}(\sigma),\tilde{\zeta}(\sigma)\right)d\sigma_ 
1
d\sigma_2  \notag \\
& -  \sum_{m=1}^3
\left\{ (x^{(j)}_m - y^{(j)}_m)(X - Z) -  (X - Y)(x^{(j)}_m - z^{(j)}_m) 
\right\}
\notag\\
& \quad  \cdot \int_0^1\int_0^1 \sigma_1 
\frac{\partial\tilde{A}_m}{\partial a_{\Lambda'}}
\left(\zeta^{(j)}(\sigma),\tilde{\zeta}(\sigma)\right)d\sigma_1
d\sigma_2.
\end{align}
So we can complete the proof of \eqref{6.13} from \eqref{6.6}.
\end{proof}
    Let's define  $\Phi_m^{(j)}(t,s;x^{(j)},y^{(j)},z^{(j)},X,Y,Z) \in R\ (m 
= 1,2,3,j =
1,2,\dots,n)$ and $\Phi_1(t,s;x,y,z,X,Y,Z) \in R^{4N}$ by
\begin{align}   \label{6.16}
     & \Phi_m^{(j)} = \bigl(z^{(j)}_m - \frac{x^{(j)}_m + y^{(j)}_m}{2}\bigr)
\notag \\
& + \frac{e_j(t-s)}{m_jc} \sum_{l=1}^3
(x^{(j)}_l - z^{(j)}_l)\int^1_0\int^1_0 \sigma_1
B_{ml}(\zeta^{(j)}(\sigma),\tilde{\zeta}(\sigma))d\sigma_1 d\sigma_2
\notag \\
&  -
\frac{e_j(t-s)}{m_jc}(X-Z)\cdot \int^1_0\int^1_0 \sigma_1
\frac{\partial\tilde{A}_m}{\partial
a_{\Lambda'}}(\zeta^{(j)}(\sigma),\tilde{\zeta}(\sigma))d\sigma_1 d\sigma_2
\notag \\
&   +
\frac{e_j(t-s)}{m_jc} \int^1_0 \tilde{A}_m(x^{(j)} - \theta (x^{(j)} - 
y^{(j)}),X- \theta
(X - Y))d\theta
   \notag \\
& + \frac{(t-s)^2}{m_j}   \int^1_0\int^1_0 \sigma_1
\partial V_1 \left(\zeta(\sigma)\right)/\partial x^{(j)}_m
d\sigma_1 d\sigma_2
\end{align}
and
\begin{align}   \label{6.17}
     & \Phi_1 =  \bigl(Z - \frac{X + Y}{2}\bigr) + \frac{(t-s)|V|}{c} 
\sum_{j=1}^n \sum_{m=1}^3
e_j(x^{(j)}_m - z^{(j)}_m) \notag \\
&\qquad  \times
   \int^1_0\int^1_0 \sigma_1
\frac{\partial\tilde{A}_m}{\partial
a_{\Lambda'}}(\zeta^{(j)}(\sigma),\tilde{\zeta}(\sigma))d\sigma_1 d\sigma_2
\notag \\
& + (t-s)^2|V|  \int^1_0\int^1_0 \sigma_1
\frac{\partial V_2}{\partial
a_{\Lambda'}}(\tilde{\zeta}(\sigma))d\sigma_1 d\sigma_2,
\end{align}
respectively.  Let $\Phi^{(j)} := 
\left(\Phi^{(j)}_1,\Phi^{(j)}_2,\Phi^{(j)}_3\right) \in
R^3.$  Then it follows from \eqref{6.13}, \eqref{6.16} and \eqref{6.17} that
\begin{align}   \label{6.18}
     & S_c(t,s;\overrightarrow{q}^{t,s}_{z,y},a_{\Lambda Z,Y}^{t,s}) -
S_c(t,s;\overrightarrow{q}^{t,s}_{z,x},a_{\Lambda Z,X}^{t,s}) \notag \\
& =
\frac{1}{t-s }\sum_{j=1}^n m_j\left(x^{(j)} - y^{(j)}\right)\cdot
\Phi^{(j)}(t,s;x^{(j)},y^{(j)},z^{(j)},X,Y,Z) \notag \\
& + \frac{1}{(t-s)|V| }(X - Y)\cdot
\Phi_1(t,s;x,y,z,X,Y,Z).
\end{align}
\section{The stability of the fundamental operator}
\begin{lem}
Let $f \in C^{1}(R^d)$ and $|\partial_x^{\alpha}f|
           \leq C_{\alpha}<x>^{-(1 + \delta_{\alpha})}$ for all $|\alpha| = 
1$, where
$\delta_{\alpha} > 0$ are constants.
Then we have : (1) $f$ is a bounded function in $R^d$. (2)  We have
\begin{align*}
      & |x - z|\left|\partial_x^{\alpha}\partial_y^{\beta}\partial_z^{\gamma}
            \int_0^1\int_0^1 \sigma_1f(z + \sigma_1(x - z)
                  + \sigma_1\sigma_2(y - x))d\sigma_1d\sigma_2\right| \\
        &  \leq C_{\alpha,\beta,\gamma},\quad  |\alpha + \beta  + \gamma|= 1,
             \  x, y, z \in R^d.
\end{align*}
\end{lem}
The proof  is easy. See the proof of Lemma 3.5 in \cite{Ichinose 1997} for 
the proof of Lemma
7.1.
\par
   We note \eqref{3.4} and \eqref{6.11}.  Then, it follows from Lemma 7.1 
that under the
assumptions of Theorem 3.1 we have
\begin{align} \label{7.1}
      & 
\Bigl|\partial_{x^{(j)}}^{\alpha}\partial_{y^{(j)}}^{\beta}\partial_{z^{(j)} 
}^{\gamma}
  \partial_X^{\alpha' }\partial_Y^{\beta'}\partial_Z^{\gamma'} 
(Z-X)\cdot
\int^1_0\int^1_0
\sigma_1
\frac{\partial\tilde{A}_m}{\partial
a_{\Lambda'}}(\zeta^{(j)}(\sigma),\tilde{\zeta}(\sigma))d\sigma_1 d\sigma_2 
\Bigr|\notag\\
        &  \leq C_{\alpha,\beta,\gamma,\alpha',\beta',\gamma'},\quad 
|\alpha + \beta  +
\gamma + \alpha'+ \beta'+ \gamma'|
\geq 0
\end{align}
for $x^{(j)},y^{(j)},z^{(j)} \in R^{3}$ and $X, Y, Z \in R^{4N}$.  In the 
same way we have
the same estimates as the above for $(x^{(j)}_l - 
z^{(j)}_l)\displaystyle{\int^1_0\int^1_0}
\sigma_1 B_{ml}(\zeta^{(j)}(\sigma),\tilde{\zeta}(\sigma))d\sigma_1 
d\sigma_2$ and
$(x^{(j)}_m - z^{(j)}_m)
\displaystyle{\int^1_0\int^1_0} \sigma_1 
\frac{\partial\tilde{A}_m}{\partial a_{\Lambda'}}
\left(\zeta^{(j)}(\sigma),\tilde{\zeta}(\sigma)\right)d\sigma_1
d\sigma_2.$  To obtain these estimates we assumed \eqref{3.6} and \eqref{3.7}.
Consequently, letting $\Theta$ be a component of $\Phi^{(j)}$ and $\Phi_1$, 
and $|\alpha +
\beta  +
\gamma + \alpha'+ \beta'+ \gamma'|
\geq 1$,
then from
\eqref{6.16} and \eqref{6.17} we obtain
\begin{align} \label{7.2}
      & \left|\partial_x^{\alpha}\partial_y^{\beta}\partial_z^{\gamma}
  \partial_X^{\alpha' }\partial_Y^{\beta'}\partial_Z^{\gamma'}
\Theta\right| \leq C_{\alpha,\beta,\gamma,\alpha',\beta',\gamma'}
\end{align}
together with \eqref{6.3} and \eqref{6.4} for $0 \leq s \leq t \leq T, x, 
y, z \in R^{3n}$ and
$X, Y, Z
\in R^{4N}$.
\begin{pro}
Under the assumptions of Theorem 3.1 we have : (1)  There exists a constant 
$\rho^* > 0$
such that the mapping : $R^{3n+4N} \ni (z,Z) \rightarrow (\xi,\Xi) = 
(\Phi,\Phi_1) :=
\left(\Phi^{(1)},\Phi^{(2)},\dotsc,\Phi^{(n)},\Phi_1\right) \in R^{3n+4N}$
is homeomorphic and
$\det \partial(\xi,\Xi)/\partial (z,Z) \geq 1/2$ for each fixed
$0 \leq t - s \leq \rho^*, x,y,X$ and $Y$.  We write its inverse mapping as
$R^{3n+4N} \ni (\xi,\Xi) \rightarrow (z,Z) =
\bigl(z(t,s;x,\xi,y,X,\Xi,Y), \\
Z(t,s;x,\xi,y,X,\Xi,Y)\bigl) \in R^{3n+4N}.$
(2) Let $\eta(t,s;x,\xi,y,X,\Xi,Y)$ be a component of $z$ and $Z$.  Then, 
letting $|\alpha +
\beta  + \gamma+ \alpha'+\beta'+\gamma'| \geq 1$, we have
\begin{align}   \label{7.3}
           & \left|\partial_{\xi}^{\alpha}\partial_x^{\beta}\partial_y^{\gamma}
\partial_{\Xi}^{\alpha'}\partial_X^{\beta'}\partial_Y^{\gamma'}
            \eta(t,s;x,\xi,y,X,\Xi,Y)\right|  \leq
C_{\alpha,\beta,\gamma,\alpha',\beta',\gamma'}
\end{align}
for $0 \leq t - s \leq
\rho^*, \ x, \xi, y \in R^{3n}$ and $X, \Xi, Y \in R^{4N}$.
\end{pro}
\begin{proof}
(1) From \eqref{6.16} and  \eqref{6.17}  we write
\begin{equation}   \label{7.4}
           \partial(\Phi,\Phi_1)/ \partial (z,Z) = I + (t - 
s)d(t,s;x,y,z,X,Y,Z),
\end{equation}
where $I$ is the identity matrix of degree
$3n+4N$. We can see as in the proof of \eqref{7.2} that each component of 
$d$ satisfies
\eqref{7.2} for all
$\alpha,\beta,\gamma,\alpha',\beta'$ and $\gamma'$.  Hence, applying 
Theorem 1.22 in
\cite{Schwartz} to the mapping : $(z,Z) \rightarrow (\Phi,\Phi_1)$, we can 
prove (1). \par
    (2)  We see
\begin{equation*}
           (\xi,\Xi) = 
\left(\Phi(t,s;x,y,z,X,Y,Z),\Phi_1(t,s;x,y,z,X,Y,Z)\right)
\end{equation*}
with $z = z(t,s;x,\xi,y,X,\Xi,Y)$ and $Z = Z(t,s;x,\xi,y,X,\Xi,Y)$.  So, 
\eqref{7.3} follows
from \eqref{7.2} and $\det \partial (\xi,\Xi)/\partial (z,Z) \geq 1/2$.
\end{proof}
\begin{rem}
Let's consider the general case of $\Lambda_2 \subseteq \Lambda_3$. Then 
from \eqref{3.4} and
\eqref{6.12} we take $\tilde{A}(x,a_{\Lambda'_2})$ and 
$B_{ml}(x,a_{\Lambda'_2})$ in
\eqref{6.16} and \eqref{6.17}. Let $\Lambda'_1$ and $\Lambda'_2$ be fixed. 
When $\Lambda'_3 =
\Lambda'_2$, we could determine
$\rho^* > 0$ from \eqref{7.4} such that we get
$\det \partial (\Phi,\Phi_1)/\partial (z,Z) \geq
1/2$ for $0 \leq t - s \leq \rho^*, x,y,z \in R^{3n}$ and $X,Y,Z \in 
R^{4N_3}$.  Let
$\Lambda'_3
\supseteq \Lambda'_2$.  Then, the direct calculations show
\[
\det \partial (\Phi,\Phi_1)/\partial (z,Z) \geq
1/2
\]
for $0 \leq t - s \leq \rho^*, x,y,z \in R^{3n}$  and $X,Y,Z
\in R^{4N_3}$ from \eqref{6.16} and \eqref{6.17} since
  $|V| \partial^2
V_2(a_{\Lambda'})/\partial (\alki)^2 = (c|k|)^2$ are positive. 
Consequently, we can see that
when
$\Lambda'_1$ and $\Lambda'_2$ are  fixed, the constant $\rho^* >0$ is taken 
independently of
$\Lambda'_3$.  See the Pauli-Fierz hamiltonian in \cite{Pauli-Fierz} for 
the condition
$\Lambda'_3
\supseteq \Lambda'_2$.
\end{rem}
\begin{thm}
Let $\rho^* > 0$ be the constant determined in Proposition 7.2. Then under 
the assumptions
of Theorem 3.1 we can find  constants
$K_a \geq 0\ (a = 0,1,2,\dotsc)$ such that
\begin{align} \label{7.5}
\left\Vert \Cts f
\right\Vert_{B^a}
    \leq \mathrm{e}^{K_a(t - s)}\Vert f \Vert_{B^a}, \  0 \leq t - s \leq \rho^*
\end{align}
for all $f(x,a_{\Lambda'}) \in B^a(R^{3n+4N})$.
\end{thm}
\begin{proof}
The definition \eqref{6.8} says
\begin{align} \label{7.6}
C(s,s) = \text{Identity}.
\end{align}
So \eqref{7.5} holds for $t = s$. \par
   Let $0 < t - s \leq \rho^*$.  We take $\chi \in \Cspace(R^{3n+4N})$ with 
compact support
such that $\chi(0) = 1$.  Let $\epsilon > 0$ and $f \in 
\Sspace(R^{3n+4N})$.  Then from
\eqref{6.8} and
\eqref{6.18} we can write
\begin{align} \label{7.7}
& \Cts^* \chi(\epsilon\cdot)^2 \Cts f =
\left\{\prod_{j=1}^n\left(\frac{m_j}{2\pi \hbar(t - s)}\right)^{3}\right\}
               \left(\frac{1}{2\pi \hbar|V|(t - s)}\right)^{4N} \iint 
f(y,Y)dydY \notag\\
& \times \iint \chi(\epsilon z,\epsilon Z)^2  \exp \left\{
i\hbar^{-1}S_c(t,s;\overrightarrow{q}^{t,s}_{z,y},a_{\Lambda Z,Y}^{t,s}) -
i\hbar^{-1}S_c(t,s;\overrightarrow{q}^{t,s}_{z,x},a_{\Lambda 
Z,X}^{t,s})\right\} dzdZ \notag\\
& = \left\{\prod_{j=1}^n\left(\frac{m_j}{2\pi \hbar(t - s)}\right)^{3}\right\}
               \left(\frac{1}{2\pi \hbar|V|(t - s)}\right)^{4N} \iint 
f(y,Y)dydY \iint
\chi(\epsilon z,\epsilon Z)^2\notag\\
& \times  \exp\left(i\sum_{j=1}^n (x^{(j)} - y^{(j)})\cdot 
\frac{m_j\Phi^{(j)}}{\hbar (t - s)}
  + i(X - Y) \cdot \frac{\Phi_1}{\hbar |V| (t - s)}\right)dzdZ.
\end{align}
We can make the change of variables : $(z,Z) \rightarrow (\xi,\Xi) = 
(\Phi,\Phi_1)$ in
\eqref{7.7} from Proposition 7.2.  Then
\begin{align*}
& \Cts^* \chi(\epsilon\cdot)^2 \Cts f =
\left\{\prod_{j=1}^n\left(\frac{m_j}{2\pi \hbar(t - s)}\right)^{3}\right\}
               \left(\frac{1}{2\pi \hbar|V|(t - s)}\right)^{4N} \notag\\
& \times  \iint f(y,Y)dydY \iint
\chi(\epsilon z,\epsilon Z)^2 \Bigg\{\exp\Bigg(i\sum_{j=1}^n (x^{(j)} - 
y^{(j)})\cdot
\frac{m_j\xi^{(j)}}{\hbar (t - s)} \notag\\
&
  + i(X - Y) \cdot \frac{\Xi}{\hbar |V|(t - s)}\Bigg)\Bigg\}\det \frac{\partial
(z,Z)}{\partial (\xi,\Xi)}d\xi d\Xi.
\end{align*}
The equation \eqref{7.4} and (2) of Proposition 7.2 show
\begin{equation} \label{7.8}
\det\frac{\partial (z,Z)}{\partial (\xi,\Xi)} = 1 + (t - 
s)h(t,s;x,\xi,y,X,\Xi,Y),
\end{equation}
where $h(t,s;x,\xi,y,X,\Xi,Y)$ satisfies \eqref{7.3} for all
$\alpha,\beta,\gamma,\alpha',\beta'$ and $\gamma'$.  Consequently from (2) 
of Proposition 7.2
we have
\begin{align} \label{7.9}
& \lim_{\epsilon \rightarrow 0}\Cts^* \chi(\epsilon\cdot)^2 \Cts f =
\left(\frac{1}{2\pi}\right)^{3n+4N}\lim_{\epsilon \rightarrow 0}\iint 
f(y,Y)dydY \iint
\chi(\epsilon z,\epsilon Z)^2 \notag\\
& \times   \bigl\{\exp\bigl(i (x- y)\cdot \gamma + i(X - Y) \cdot \Gamma 
\bigr)\bigr\}
\det \frac{\partial
(z,Z)}{\partial (\xi,\Xi)}d\gamma d\Gamma \notag \\
& = f(x,X) + (t - s) \left(\frac{1}{2\pi}\right)^{3n+4N} \text{Os} -\iiiint 
\left\{\exp\bigl(i
(x- y)\cdot
\gamma + i(X - Y) \cdot \Gamma \bigr) \right\}
  \notag\\
& \times
  h(t,s;x,\xi,y,X,\Xi,Y)f(y,Y)dydYd\gamma d\Gamma,
\end{align}
where $\xi^{(j)} = \hbar (t - s)\gamma^{(j)}/m_j\ (j= 1,2,\dotsc,n)$ and 
$\Xi = \hbar|V|(t -
s)\Gamma$. We note that the second term on the rhs of \eqref{7.9} is a 
pseudo-differential
operator.  So, applying the Calder\'on-Vaillancourt theorem 
(\cite{Calderon-Vaillancourt}),
we obtain
\begin{align*}
& \lim_{\epsilon \rightarrow 0} \|\chi(\epsilon\cdot)\Cts f \|^2
=\lim_{\epsilon \rightarrow 0} \bigl(\Cts^* \chi(\epsilon\cdot)^2 \Cts f, f 
\bigr)  \\
& = \bigl(\lim_{\epsilon \rightarrow 0}\Cts^* \chi(\epsilon\cdot)^2 \Cts f, 
f \bigr)  \leq (1 +
2K_0(t-s))\|f\|^2 \\ & \leq \mathrm{e}^{2K_0(t - s)} \|f\|^2
\end{align*}
with a constant $K_0 \geq 0$.  Hence we get \eqref{7.5}
with $a=0$ by Fatou's lemma.  \par
Let $p(x,w,X,W)$ be a $\Cspace$ function satisfying
\eqref{6.7} with an integer
$M \geq 0$.  Then we obtain
\begin{equation} \label{7.10}
\|P(t,s)f\| \leq \text{Const.} \|f\|_{B^{M}}
\end{equation}
as in the proof of \eqref{7.5} with $a=0$.  See the proof of Proposition 
4.3  in \cite{Ichinose
1999} for  further details. \par
   Let's remember the expression \eqref{6.9} of $\Cts f$.
Set $\zeta := (x,X)$ and let $\kappa = 
(\kappa_1,\kappa_2,\dotsc,\kappa_{3n+4N})$ be an
arbitrary multi-index.  Then we can see that
$\partial_{\zeta}^{\kappa}(\Cts f) -
\Cts(\partial_{\zeta}^{\kappa}f)$ and
$\zeta^{\kappa}(\Cts f) - \Cts(\zeta^{\kappa}f)$ are written in the form
\begin{align}  \label{7.11}
   & (t - s)\sum_{|\gamma| \leq 
|\kappa|}\tilde{P}_{\gamma}(t,s)(\partial_{\zeta}^{\gamma}f)
     := (t - s)\sum_{|\gamma| \leq |\kappa|} \left( \prod_{j=1}^n 
\sqrt{\frac{m_j}{2\pi
i\hbar}}^{\ 3} \right)\sqrt{\frac{1}{2\pi i\hbar|V|}}^{\ 4N}
  \notag \\
       & \quad \times \text{Os}-\iint
\bigl(\exp 
i\hbar^{-1}\phi(t,s;x,w,X,W)\bigr)p_{\gamma}(t,s;x,\sqrt{\rho}w,X,\sqrt{\rho 
}W )
\notag
\\ & \qquad \times (\partial_{\zeta}^{\gamma}f)(x -\sqrt{\rho}w, X 
-\sqrt{\rho}W)dw dW
\end{align}
respectively, where $p_{\gamma}(t,s;x,w,X,W)$ satisfies \eqref{6.7} with $M 
= |\kappa| -
|\gamma|$ for all $\alpha, \beta, \alpha'$ and $\beta'$.  We can prove 
these results
  by induction with respect to
$|\kappa|$, using
$\hbar \partial_{w^{(j)}} \mathrm{e}^{im_j\hbar^{-1}|w^{(j)}|^2/2} = im_j 
w^{(j)}
\mathrm{e}^{im_j\hbar^{-1}|w^{(j)}|^2/2}$, $\hbar\partial_{W}
\mathrm{e}^{i\hbar^{-1}|W|^2/(2|V|)} =(i W/|V|)
  \mathrm{e}^{i\hbar^{-1}|W|^2/(2|V|)}$ and the integration by parts in 
\eqref{6.9}.  See the
proof of Lemma 3.2 in \cite{Ichinose 2003} for further details.  \par
   Let $|\kappa| = a\ (a = 0,1,2,\dotsc)$.  Then we have
\begin{equation*}
\Vert\partial_{\zeta}^{\kappa}(\Cts f)\Vert \leq \Vert\Cts 
(\partial_{\zeta}^{\kappa}f)\Vert
+
   (t - s)\sum_{|\gamma| \leq a}\Vert 
\tilde{P}_{\gamma}(t,s)(\partial_{\zeta}^{\gamma}f)\Vert.
\end{equation*}
Applying \eqref{7.5} with $a = 0$ and \eqref{7.10} to the rhs above, we get
\begin{align*}
\Vert\partial_{\zeta}^{\kappa}(\Cts f)\Vert
& \leq e^{K_0(t -s)}\Vert\partial_{\zeta}^{\kappa}f\Vert +
   \text{Const.}(t - s)\sum_{|\gamma| \leq a}\Vert
\partial_{\zeta}^{\gamma}f\Vert_{B^{a-|\gamma|}}.
\end{align*}
We know from Lemma 2.3 with $s = 1$ and $a = b$ in \cite{Ichinose 1995} 
that there exist a
constant $\mu_a \geq 0$ and $\lambda_a(\zeta,\eta)$ satisfying
\begin{equation} \label{7.12}
   |\partial^{\alpha}_{\eta}\partial^{\beta}_{\zeta} \lambda_a(\zeta,\eta)|
      \leq C_{\alpha,\beta}<\zeta;\eta>^{-a}
\end{equation}
for all $\alpha$ and $\beta$, and
\begin{equation} \label{7.13}
    \Lambda_a(\zeta,D_{\zeta}) =
   \left(\mu_a +   <\zeta>^{a} + <D_{\zeta}>^{a}\right)^{-1}
\end{equation}
on $\Sspace$, where $\Lambda_a(\zeta,D_{\zeta})$ is the pseudo-differential 
operator with
symbol $\lambda_a(\zeta,\eta)$.  So, using Lemma 2.4 in \cite{Ichinose 
1995} and the
Calder\'on-Vaillancourt theorem,  we have
\begin{align} \label{7.14}
& \Vert \partial_{\zeta}^{\gamma}f\Vert_{B^{a-|\gamma|}}
  \leq
\text{Const.}  \Vert
  \left(\mu_{a-|\gamma|} +   <\zeta>^{a-|\gamma|} +
<D_{\zeta}>^{a-|\gamma|}\right)\partial_{\zeta}^{\gamma}f\Vert \notag \\
& =
\text{Const.}  \Vert
  \left\{\left(\mu_{a-|\gamma|} +   <\zeta>^{a-|\gamma|} +
<D_{\zeta}>^{a-|\gamma|}\right)\partial_{\zeta}^{\gamma}\Lambda_a
\right\}\bigl(\mu_{a} +   <\zeta>^{a} \notag \\
&\qquad  + <D_{\zeta}>^{a}\bigr)f\Vert \leq
\text{Const.}
\Vert f\Vert_{B^{a}}.
\end{align}
Hence we get
\begin{align} \label{7.15}
\Vert\partial_{\zeta}^{\kappa}(\Cts f)\Vert
& \leq e^{K_0(t -s)}\Vert\partial_{\zeta}^{\kappa}f\Vert +
   \text{Const.}(t - s)\Vert
f\Vert_{B^{a}}.
\end{align}
In the same way we get
\begin{align} \label{7.16}
\Vert\zeta^{\kappa}(\Cts f)\Vert
& \leq e^{K_0(t -s)}\Vert\zeta^{\kappa}f\Vert +
   \text{Const.}(t - s)\Vert
f\Vert_{B^{a}}.
\end{align}
Thus we obtain
\begin{align*}
\Vert \Cts f\Vert_{B^{a}}
& \leq e^{K_0(t -s)}\Vert f\Vert_{B^{a}} +
   \text{Const.}(t - s)\Vert
f\Vert_{B^{a}} \\
& \leq e^{K_a(t -s)}\Vert f\Vert_{B^{a}}.
\end{align*}
This completes the proof of Theorem 7.3.
\end{proof}
\begin{pro}
Let $0 \leq t - s \leq \rho^*$ and $p(x,w,X,W)$ satisfy \eqref{6.7} with an 
integer $M
\geq 0$.  Then $P(t,s)$ is a continuous operator from $B^a\ (a = 
0,1,2,\dotsc)$ into $B^{a+M}$.
\end{pro}
\begin{proof}
  Let $\zeta = (x,X)$ and $f \in \Sspace(R^{3n+4N})$. We also use 
\eqref{6.9} as in the proof of
Theorem 7.3.  Then we have
\[
\partial_{\zeta}^{\kappa}P(t,s)f = \sum_{\gamma \leq \kappa}P_{\gamma}(t,s)
\partial_{\zeta}^{\gamma}f,
\]
where $\gamma \leq \kappa$ denotes $\gamma_j \leq \kappa_j$ for all $j$ and
$p_{\gamma}(t,s;x,w,X,W)$ satisfy \eqref{6.7} with $M + |\kappa|-|\gamma|$ 
as $M$. Using
$\zeta = (x,X) = (x - \sqrt{\rho}w,X - \sqrt{\rho}W) + \sqrt{\rho}(w,W)$, 
we also have
\[
\zeta^{\kappa}P(t,s)f = \sum_{\gamma \leq \kappa}Q_{\gamma}(t,s)
\zeta^{\gamma}f,
\]
where $q_{\gamma}(t,s;x,w,X,W)$ satisfy \eqref{6.7} with $M + 
|\kappa|-|\gamma|$ as $M$.
Hence   from \eqref{7.10} and \eqref{7.14} we see
\begin{align} \label{7.17}
\Vert P(t,s) f\Vert_{B^{a}}
& = \Vert P(t,s) f\Vert + \sum_{|\kappa| = a}\left(\Vert 
\zeta^{\kappa}P(t,s) f\Vert +
\Vert \partial_{\zeta}^{\kappa}P(t,s) f\Vert \right) \notag
    \\
& \leq \text{Const.}\Vert f\Vert_{B^{a+M}}.
\end{align}
So we could complete the proof.
\end{proof}
\section{The consistency of the fundamental operator}
    Let $\Cts$ and $H(t)$ be the fundamental operator defined in \S 6 and 
the operator defined
by \eqref{3.10} with $a_{\Lambda'} = a_{\Lambda'_2} = X$, respectively.
\begin{thm}
Under the assumptions of Theorem 3.1 there exist integers $M \geq 0, M' 
\geq 0, \Cspace$
functions $r(t,s;x,w,X,W)$ and $r'(t,s;x,w,X,W)$
in  $0 \leq s \leq t \leq T$, $(x,w) \in R^{6n}$ and $(X,W) \in R^{8N}$ 
satisfying
\eqref{6.7} for all $\alpha, \beta, \alpha'$ and $\beta'$, respectively 
such that
\begin{equation}  \label{8.1}
       \left(i\hbar\frac{\partial}{\partial t} - H(t)\right)\Cts f
        = \sqrt{t - s}R(t,s)f
\end{equation}
and
\begin{equation}  \label{8.2}
       i\hbar\frac{\partial}{\partial s}\Cts f + \Cts H(s) f
        = \sqrt{t - s}R'(t,s)f,
\end{equation}
where $R(t,s)$ and $R'(t,s)$ are the operators defined by \eqref{6.8}.
\end{thm}
\begin{proof}
In this proof we write $x$ and $y$ as $\xvec$ and $\yvec$, respectively. 
Let $x$ denote
variables in $R^3$. It follows from \eqref{3.10}, \eqref{6.6} and 
\eqref{6.8} that the direct
calculations show
\begin{align}  \label{8.3}
       & \left(i\hbar\frac{\partial}{\partial t} - H(t)\right)\Cts f
        = - \left(\prod_{j=1}^n\sqrt{\frac{m_j}{2\pi i\hbar(t - s)}}^{\ 
3}\right)
               \sqrt{\frac{1}{2\pi i\hbar|V|(t - s)}}^{\ 4N} \notag \\
& \times \iint \left( \exp 
i\hbar^{-1}S_c(t,s;\overrightarrow{q}^{t,s}_{\xvec,\yvec},a_{\Lambda
X,Y}^{t,s}\right)
\Bigl\{r_1(t,s;\xvec,\yvec,X,Y)  \notag\\
&\qquad + \frac{i\hbar}{2}r_2(t,s;\xvec,\yvec,X,Y)\Bigr\}f(\yvec,Y)d\yvec dY
\end{align}
by means of \eqref{6.3} and \eqref{6.4},
where
\begin{align}  \label{8.4}
        r_1(t,s;\xvec,\yvec,X,Y) & =\partial_t\,
S_c(t,s;\overrightarrow{q}^{t,s}_{\xvec,\yvec},a_{\Lambda X,Y}^{t,s}) +
\sum_{j=1}^n \frac{1}{2m_j}
            \bigl|\partial_{x^{(j)}}\, S_c - 
\frac{e_j}{c}\tilde{A}(x^{(j)},X)\bigr|^2
\notag
\\  & + V_1(\xvec) + \frac{|V|}{2}\left| \partial_X S_c\right|^2 + V_2(X)
\end{align}
and
\begin{align}  \label{8.5}
         r_2 & = \frac{3n + 4N}{t-s} - \sum_{j=1}^n 
\frac{1}{m_j}\Delta_{x^{(j)}}\,S_c
\notag\\
& + \frac{1}{c}\sum_{j=1}^n \frac{e_j}{m_j}(\nabla_{x} \cdot 
\tilde{A})(x^{(j)},X)
- |V|\Delta_X S_c, \quad x \in R^3
\end{align}
  (cf. the proof of Proposition 2.3 in \cite{Ichinose 1997}). \par
   Set $\rho = t - s$.  From \eqref{6.6} we can write
\begin{align}  \label{8.6}
       & \partial_{x^{(j)}} S_c - \frac{e_j}{c}\tilde{A}(x^{(j)},X)
= \frac{m_j(x^{(j)}- y^{(j)})}{\rho} \notag \\
        &   + \frac{e_j}{c} \int^1_0 \left\{ \tilde{A}\left(x^{(j)} - 
\theta(x^{(j)} -
y^{(j)}), X - \theta (X - Y)\right)
               - \tilde{A}(x^{(j)},X)\right\}d\theta  \notag \\
             & +  \frac{e_j}{c}\sum_{l=1}^3 (x^{(j)}_l - y^{(j)}_l) \int^1_0
(1 - \theta )
                 \frac{\partial \tilde{A}_l}{\partial x}\left(x^{(j)} -
\theta(x^{(j)} - y^{(j)}), X - \theta(X - Y)\right)d\theta \notag\\
&
          - \rho \int^1_0 (1 - \theta )\frac{\partial V_1}{\partial x^{(j)}}
                       (\xvec - \theta(\xvec - \yvec))d\theta \notag\\
       & = \frac{m_j(x^{(j)}- y^{(j)})}{\rho} - \frac{e_j}{2c}\sum_{m=1}^3 
(x^{(j)}_m -
y^{(j)}_m) \frac{\partial \tilde{A}}{\partial x_m}(x^{(j)},X) \notag \\
& - \frac{e_j}{2c}\sum_{m=1}^{4N}  (X_m - Y_m) \frac{\partial \tilde{A}}
{\partial X_m}(x^{(j)},X) + \frac{e_j}{2c}\sum_{l=1}^3  (x^{(j)}_l -
y^{(j)}_l) \frac{\partial \tilde{A}_l}{\partial x}(x^{(j)},X) \notag \\
& + \rho q_1(t,s;\xvec,\frac{\xvec - \yvec}{\sqrt{\rho}},X, \frac{X - 
Y}{\sqrt{\rho}})
\end{align}
and
\begin{align}  \label{8.7}
       & \partial_{X} S_c
= \frac{X - Y}{|V|\rho} - \rho \int^1_0 (1 - \theta )\frac{\partial 
V_2}{\partial X}
                       (X - \theta(X - Y))d\theta \notag \\
             & +  \frac{1}{c}\sum_{j=1}^n\sum_{l=1}^3 e_j(x^{(j)}_l - 
y^{(j)}_l) \int^1_0
(1 - \theta )
                 \frac{\partial \tilde{A}_l}{\partial X}\left(x^{(j)} -
\theta(x^{(j)} - y^{(j)}), X - \theta(X - Y)\right)d\theta \notag\\
& = \frac{X - Y}{|V|\rho} + \frac{1}{2c}\sum_{j=1}^n\sum_{l=1}^3 
e_j(x^{(j)}_l - y^{(j)}_l)
                 \frac{\partial \tilde{A}_l}{\partial X}\left(x^{(j)}, X \right)
+ \rho q_2(t,s;\xvec,\frac{\xvec - \yvec}{\sqrt{\rho}},X, \frac{X - 
Y}{\sqrt{\rho}}).
\end{align}
It holds that
\begin{align}  \label{8.8}
       & - \sum_{k,m=1}^3\left(x^{(j)}_k- y^{(j)}_k\right)\left(x^{(j)}_m -
y^{(j)}_m\right) \frac{\partial \tilde{A}_k}{\partial x_m}(x^{(j)},X) \notag \\
& + \sum_{k,l=1}^3\left(x^{(j)}_k- y^{(j)}_k\right)\left(x^{(j)}_l -
y^{(j)}_l\right) \frac{\partial \tilde{A}_l}{\partial x_k}(x^{(j)},X) = 0.
\end{align}
The equations \eqref{8.6} - \eqref{8.8} show
\begin{align}  \label{8.9}
         & \sum_{j=1}^n \frac{1}{2m_j} \bigl|\partial_{x^{(j)}}\, S_c -
\frac{e_j}{c}\tilde{A}(x^{(j)},X)\bigl|^2 + \frac{|V|}{2}\left| \partial_X 
S_c\right|^2
\notag
\\  & = \frac{1}{2\rho^2} \sum_{j=1}^n m_j\left| x^{(j)} - y^{(j)} \right|^2
+ \frac{|X - Y|^2}{2|V|\rho^2} + \sqrt{\rho} q_3(t,s;\xvec,\frac{\xvec - 
\yvec}{\sqrt{\rho}},X,
\frac{X - Y}{\sqrt{\rho}}).
\end{align}
 From \eqref{6.6} we also have
\begin{align}  \label{8.10}
          \partial_t\, S_c(t,s;& \qvec_{\xvec,\yvec}^{t,s},a_{\Lambda 
X,Y}^{t,s})  =
- \frac{1}{2\rho^2} \sum_{j=1}^n m_j\left| x^{(j)} - y^{(j)} \right|^2 - 
V_1(\xvec)
- \frac{|X - Y|^2}{2|V|\rho^2} \notag \\
& -V_2(X) + \sqrt{\rho} q_4(t,s;\xvec,\frac{\xvec - \yvec}{\sqrt{\rho}},X, 
\frac{X -
Y}{\sqrt{\rho}}).
\end{align}
Hence together with \eqref{8.4} we obtain
\begin{equation}  \label{8.11}
r_1(t,s;\xvec,\yvec,X,Y)  = \sqrt{\rho} q_5(t,s;\xvec,\frac{\xvec - 
\yvec}{\sqrt{\rho}},X,
\frac{X - Y}{\sqrt{\rho}}).
\end{equation} \par
   From \eqref{6.6} or \eqref{8.6} - \eqref{8.7} the same arguments as for 
$r_1$ show
\begin{align}  \label{8.12}
          & \sum_{j=1}^n \frac{1}{m_j}\Delta_{x^{(j)}}\,S_c   + |V|\Delta_X 
S_c= \frac{3n +
4N}{\rho} +  \frac{2}{c}\sum_{j=1}^n \frac{e_j}{m_j}\int^1_0 (1 - \theta)
\notag\\
& \times \left(\nabla_{x} \cdot
A\right) \left(x^{(j)} -
\theta(x^{(j)} - y^{(j)}), X - \theta (X - Y)\right) d\theta \notag \\
& + \sqrt{\rho} q_6(t,s;\xvec,\frac{\xvec - \yvec}{\sqrt{\rho}},X, \frac{X -
Y}{\sqrt{\rho}})  = \frac{3n + 4N}{\rho} \notag \\
& + \frac{1}{c} \sum_{j=1}^n\frac{e_j}{m_j}\left(\nabla_{x} \cdot
\tilde{A}\right)(x^{(j)},X) + \sqrt{\rho} q_7(t,s;\xvec,\frac{\xvec - 
\yvec}{\sqrt{\rho}},X,
\frac{X - Y}{\sqrt{\rho}}).
\end{align}
Hence together with \eqref{8.5} we get
\begin{equation}  \label{8.13}
r_2(t,s;\xvec,\yvec,X,Y)  = -\sqrt{\rho} q_7(t,s;\xvec,\frac{\xvec - 
\yvec}{\sqrt{\rho}},X,
\frac{X - Y}{\sqrt{\rho}}).
\end{equation}
Thus we could complete the proof of \eqref{8.1} from \eqref{8.3}, 
\eqref{8.11} and
\eqref{8.13}.
\par
    We consider \eqref{8.2}.  By direct calculations we  see that the lhs of 
\eqref{8.2}
is equal to
\begin{align}  \label{8.14}
       &  - \left(\prod_{j=1}^n\sqrt{\frac{m_j}{2\pi i\hbar(t - s)}}^{\ 
3}\right)
               \sqrt{\frac{1}{2\pi i\hbar|V|(t - s)}}^{\ 4N} \notag \\
& \times \iint \left( \exp 
i\hbar^{-1}S_c(t,s;\overrightarrow{q}^{t,s}_{\xvec,\yvec},a_{\Lambda
X,Y}^{t,s}\right)
\Bigl\{r'_1(t,s;\xvec,\yvec,X,Y)  \notag\\
&\qquad + \frac{i\hbar}{2}r'_2(t,s;\xvec,\yvec,X,Y)\Bigr\}f(\yvec,Y)d\yvec dY,
\end{align}
where
\begin{align}  \label{8.15}
        r'_1(t,s;\xvec,\yvec,X,Y) & =\partial_s\,
S_c(t,s;\overrightarrow{q}^{t,s}_{\xvec,\yvec},a_{\Lambda X,Y}^{t,s}) -
\sum_{j=1}^n \frac{1}{2m_j}
            \bigl|\partial_{y^{(j)}}\, S_c + 
\frac{e_j}{c}\tilde{A}(y^{(j)},Y)\bigr|^2
\notag
\\  & - V_1(\yvec) - \frac{|V|}{2}\left| \partial_Y S_c\right|^2 - V_2(Y)
\end{align}
and
\begin{align}  \label{8.16}
         r'_2 & = - \frac{3n + 4N}{t-s} + \sum_{j=1}^n 
\frac{1}{m_j}\Delta_{y^{(j)}}\,S_c
\notag\\
& + \frac{1}{c}\sum_{j=1}^n \frac{e_j}{m_j}(\nabla_{x} \cdot 
\tilde{A})(y^{(j)},Y)
+ |V|\Delta_Y S_c.
\end{align}
Consequently we can prove \eqref{8.2} as in the proof of \eqref{8.1}.
\end{proof}
\section{The proofs of the main results}
We first prove Theorem 3.1. Let $\rho^* > 0$ be the constant determined in 
Proposition 7.2
and $\chi
\in
\Cspace(R^{3n+4N})$  with compact support such that $\chi (0) = 1$.    We 
consider bounded
operators $K_j$ and
$K'_j\ (j = 1,2,\dotsc,\nu)$ on
$B^a(R^{3n+4N})$.  Then, it holds for $f \in B^a(R^{3n+4N})$ that
\begin{align} \label{9.1}
&
K_{\nu} \chi(\epsilon\cdot) K_{\nu - 1}  \chi(\epsilon\cdot)
\cdots \chi(\epsilon\cdot) K_1\chi(\epsilon\cdot) f
- K'_{\nu} K'_{\nu - 1}  \cdots  K'_1f
     \notag \\
  & = \sum_{j=1}^{\nu} K_{\nu} \chi(\epsilon\cdot) \cdots 
\chi(\epsilon\cdot)K_{j+1}
\chi(\epsilon\cdot)\left(K_{j} - K'_{j}\right)K'_{j-1} \cdots  K'_1f \notag\\
& \quad + \sum_{j=0}^{\nu-1} K_{\nu} \chi(\epsilon\cdot) \cdots 
\chi(\epsilon\cdot)K_{j+1}
\left(\chi(\epsilon\cdot) - 1\right)K'_{j} \cdots  K'_1f.
\end{align}
\par
   Noting \eqref{6.1} and \eqref{6.2}, from \eqref{3.5} we have
\begin{equation*}
S_c(T,0;\qvec_{\Delta},a_{\Lambda\Delta}) = \sum_{l=1}^{\nu}
S_c 
\left(\tau_l,\tau_{l-1};\qvec_{x^{(l)},x^{(l-1)}}^{\tau_l,\tau_{l-1}},
a_{\Lambda
X^{(l)},X^{(l-1)}}^{\tau_l,\tau_{l-1}}\right),
\end{equation*}
where  $X^{(l)} = a_{\Lambda'}^{(l)}\ (l = 1,2,\dotsc,\nu-1)$ and $X^{(\nu)}  =
a_{\Lambda'}$. So, \eqref{3.8} is written as
\begin{equation*}
     \lim_{\epsilon\rightarrow 0}\  {\cal C}(T,\tau_{\nu-1})\chi(\epsilon\cdot)
     {\cal C}(\tau_{\nu-1},\tau_{\nu-2})\chi(\epsilon\cdot)\cdots{\cal
C}(\tau_2,\tau_1)\chi(\epsilon\cdot)  {\cal C}(\tau_1,0)\chi(\epsilon\cdot)f
\end{equation*}
for $f \in B^a(R^{3n+4N})$. Let $f \in B^a(R^{3n+4N})$ and $|\Delta|\leq 
\rho^*.$
We can easily see
\begin{equation*}
\sup_{0 < \epsilon \leq 1} \Vert \chi(\epsilon\cdot)f\Vert_{B^a}
\leq \text{Const.}\Vert f\Vert_{B^a}
\end{equation*}
and
\begin{equation*}
   \lim_{\epsilon\rightarrow 0}\Vert (\chi(\epsilon\cdot) - 1)f\Vert_{B^a}
= 0.
\end{equation*}
Consequently, using Theorem 7.3 and \eqref{9.1}, we can see that there exists
\eqref{3.8} in
$B^a$, which is written as
\begin{equation} \label{9.2}
{\cal C}(T,\tau_{\nu-1})
     {\cal C}(\tau_{\nu-1},\tau_{\nu-2})\cdots{\cal
C}(\tau_2,\tau_1)  {\cal C}(\tau_1,0)f
\left( = \Cdelta(T,0)f \right).
\end{equation}
  We also see from Remark 3.4 that there exists \eqref{3.8}
in $ \Sspace$. \par
   Let $0 \leq s \leq t \leq T$.  For a subdivision $\Delta$ of $[0,T]$ we 
can find $j$ and $l$
such that
$j \leq l, \tau_{j-1}  < s \leq \tau_j$ and $\tau_{l-1}  < t \leq \tau_l$, 
where we take
$j=1$ for $s = 0$.
Then we define
\begin{align} \label{9.3}
\Cdelta(t,s)f
   = & \lim_{\epsilon \rightarrow 0}{\cal C}(t,\tau_{l-1})\chi(\epsilon
\cdot){\cal C}(\tau_{l-1},\tau_{l-2})\chi(\epsilon \cdot) \notag \\
& \cdots  \chi(\epsilon
\cdot){\cal C}(\tau_{j+1},\tau_j)
\chi(\epsilon \cdot) {\cal C}(\tau_j,s)\chi(\epsilon \cdot)f
\end{align}
for $f \in B^a$ as was stated in Remark 3.2. Then we have
\begin{equation*}
\Cdelta(t,s)f
   =  {\cal C}(t,\tau_{l-1}){\cal C}(\tau_{l-1},\tau_{l-2})
  \cdots  {\cal C}(\tau_{j+1},\tau_j)
  {\cal C}(\tau_j,s)f
\end{equation*}
as in the proof of \eqref{9.2}.
Consequently, from
\eqref{7.5} we have
\begin{equation} \label{9.4}
\Vert \Cdelta(t,s)f\Vert_{B^a} \leq \mathrm{e}^{K_a(t - s)}\Vert f 
\Vert_{B^a}\quad (a =
0,1,2,\dotsc)
\end{equation}
for $|\Delta| \leq \rho^*$ under the assumptions of Theorem 3.1.
\begin{pro}
Let $|\Delta| \leq \rho^*$.  Then,
under the assumptions of Theorem 3.1 we can find an integer $M \geq 2$ such 
that
\begin{equation}  \label{9.5}
   \Vert\Cdelta(t,s)f - \Cdelta(t',s')f \Vert_{B^a}
\leq C_a(|t - t'| + |s - s'|)\Vert f\Vert_{B^{a+M}}
\end{equation}
for  $\, 0 \leq s \leq t \leq T, 0 \leq s' \leq t' \leq T$ and $a =
0,1,2,\dotsc$.
\end{pro}
\begin{proof}
Let $R(t,s)$ and $R'(t,s)$ be the operators defined by \eqref{8.1} and 
\eqref{8.2},
respectively.  We determine $M$ in Proposition 9.1 by $\max\ (M, M',2)$ for 
$M$ and $M'$ in
Theorem 8.1.  We can easily see
\begin{equation}  \label{9.6}
       i\hbar\bigl(\Cts f - {\cal C}(t',s)f\bigr) = \int_{t'}^t
\bigl(H(\theta){\cal C}(\theta,s)f + \sqrt{\theta - s}R(\theta,s)f\bigr)d\theta
\end{equation}
from \eqref{8.1} for $s \leq t' \leq t \leq T$.   Let $\tau_j < t
\leq \tau_{j+1}$ and $\tau_k < t' \leq \tau_{k+1}$.  So  $j \geq k$ holds. 
Using the
equation just after
\eqref{9.3} and \eqref{9.6}, we get
\begin{align}  \label{9.7}
      & \quad i\hbar\bigl(\Cdelta (t,s)f - \Cdelta (t',s)f\bigr)
        \notag\\
       & =\int_{t'}^{t}H(\theta)\Cdelta (\theta,s)f d\theta +
\int_{\tau_j}^{t}\sqrt{\theta - \tau_j}R(\theta,\tau_j)d\theta\,\Cdelta 
(\tau_j,s)f
\notag\\
& \quad + \sum_{l=1}^{j-k-1}\int_{\tau_{j-l}}^{\tau_{j-l+1}}\sqrt{\theta - 
\tau_{j-l}}
R(\theta,\tau_{j-l})d\theta\,\Cdelta (\tau_{j-l},s) f
\notag\\
& \quad + \int_{t'}^{\tau_{k+1}}\sqrt{\theta - 
\tau_k}R(\theta,\tau_k)d\theta\,\Cdelta
(\tau_k,s)f.
\end{align}
See the proof of Theorem 4.2 in \cite{Ichinose 2003} for further details. \par
    As in the proof of \eqref{7.14} we see
\begin{equation} \label{9.8}
   \Vert H(t)f\Vert_{B^a}
\leq \text{Const.}\Vert f\Vert_{B^{a+M}}
\end{equation}
from \eqref{3.10} because of $M \geq 2$.  We also see
\begin{equation} \label{9.9}
   \Vert R(t,s)f\Vert_{B^a}
\leq \text{Const.}\Vert f\Vert_{B^{a+M}}
\end{equation}
from Proposition 7.4 for $0 \leq t - s \leq \rho^*$.  Consequently,
\eqref{9.4} and
\eqref{9.7} show
\begin{align*}
      & \quad \hbar \Vert\Cdelta (t,s)f - \Cdelta (t',s)f\Vert_{B^a}
        \notag\\
& \leq \text{Const.}
e^{K_{a+M}T}(1 + \sqrt{\rho^*})|t - t'|\Vert f\Vert_{B^{a+M}}
\end{align*}
for $0 \leq s \leq t' \leq t \leq T$.  The inequality above holds
for $0 \leq s \leq t', t \leq T$.  In the same way we get
\begin{align*}
      & \quad \hbar \Vert\Cdelta (t,s)f - \Cdelta (t,s')f\Vert_{B^a}
        \notag\\
& \leq \text{Const.}
e^{K_{a+M}T}(1 + \sqrt{\rho^*})|s - s'|\Vert f\Vert_{B^{a+M}}
\end{align*}
for $0 \leq s, s'  \leq t \leq T$. Hence, we can complete the
proof of Proposition 9.1.
\end{proof} \par
    Let $M \geq 2$ be the integer determined in Proposition 9.1.  Let
$\{\Delta_j\}_{j=1}^{\infty}$ be a family of subdivisions of $[0,T]$ such that
$|\Delta_j| \leq \rho^*$ and $\lim_{j\rightarrow \infty}|\Delta_j| = 0$. 
Take an arbitrary
$f \in B^{a + 2M}\ (a = 0,1,2,\dotsc)$.  Then we see from \eqref{9.4} and 
\eqref{9.5} that
$\{{\cal C}_{\Delta_{j}}(t,s)f\}_{j=1}^{\infty}$
  is uniformly bounded as a family of
$B^{a+2M}$-valued continuous functions and equicontinuous as a family of
$B^{a+M}$-valued functions in $0 \leq s \leq t \leq T$, respectively.  It 
follows from the
Rellich criterion (cf. Theorem XIII. 65 in \cite{Reed-Simon}) that the 
embedding map from
$B^M$ into
$L^2$ is compact.   So is the embedding map from
$B^{a + 2M}$ into $B^{a + M}$ from \eqref{7.12}, \eqref{7.13} and Lemma 2.5 in
\cite{Ichinose 1995} with $a = b = 1$. Consequently, from Ascoli-Arzel\`{a} 
theorem we can
find a subsequence $\{\Delta_{j_k}\}_{k=1}^{\infty}$, which may depend on 
$f$, such that
${\cal C}_{\Delta_{j_k}}(t,s)f$    converges in $B^{a+M}$ uniformly
in $0 \leq s \leq t \leq T$  as $k \rightarrow \infty$.  Since ${\cal 
C}_{\Delta_{j}}(s,s)f
= f$ follows from Lemma 6.1, so \eqref{9.7} - \eqref{9.9} show that 
$\lim_{k \rightarrow
\infty}{\cal C}_{\Delta_{j_k}}(t,s)f = U(t,s)f$, where $U(t,s)f$ is 
$B^{a+M}$-valued
continuous and $B^{a}$-valued continuously differentiable function in $0 
\leq s \leq t \leq
T$ satisfying
\eqref{3.9} with $u(s) = f$.  Noting $M \geq 2$, we can easily see from the 
energy inequality
that the solutions to
\eqref{3.9} are unique in the class of $B^{a+M}$-valued continuous and 
$B^{a}$-valued
continuously differentiable functions.  Hence, we can see that
$\Cdelta (t,s)f$ converges to $U(t,s)f$ in $B^{a+M}$ uniformly in
$0
\leq s
\leq t
\leq T$ as $|\Delta| \rightarrow 0$. \par
   Take an arbitrary $f \in B^a$.  Let $\Delta$ and $\Delta'$ be
  subdivisions such that $|\Delta|\leq \rho^*$ and $|\Delta'| \leq \rho^*$. For
any $\epsilon > 0$ we can take a $g \in B^{a + 2M}$ such that $\Vert g - 
f\Vert_{B^a}
< \epsilon$. Then from \eqref{9.4} we have
\begin{align*}
   & \Vert {\cal C}_{\Delta}(t,s)f - {\cal C}_{\Delta'}(t,s)f\Vert_{B^a}
  \leq \Vert {\cal C}_{\Delta}(t,s)g - {\cal C}_{\Delta'}(t,s)g\Vert_{B^a} \\
&\quad  +
\Vert {\cal C}_{\Delta}(t,s)(f - g) \Vert_{B^a} + \Vert {\cal 
C}_{\Delta'}(t,s)(f - g)
\Vert_{B^a}
      \\
& \leq \Vert {\cal C}_{\Delta}(t,s)g - {\cal 
C}_{\Delta'}(t,s)g\Vert_{B^{a+M}} +
2\mathrm{e}^{K_aT}\epsilon.
\end{align*}
So,
\begin{equation} \label{9.10}
     \overline{\lim}_{|\Delta|,|\Delta'|\rightarrow 0}\, \max_{0\leq s \leq 
t \leq T}
\Vert {\cal C}_{\Delta}(t,s)f - {\cal C}_{\Delta'}(t,s)f\Vert_{B^a} \leq
2\mathrm{e}^{K_aT}\epsilon.
\end{equation}
Hence, we can see that  ${\cal C}_{\Delta}(t,s)f$ converges in $B^a$
uniformly in $0\leq s \leq t \leq T$ as $|\Delta|\rightarrow 0$.  We write 
this limit as
$W(t,s)f$.  \par
     Let $f \in B^a$.
  Take
$f_j \in B^{a+M}$ such that $\lim_{j\rightarrow \infty}f_j = f$ in $B^a$. 
 From \eqref{9.7} we
have
\[
i\hbar \left(W(t,s)f_j - f_j\right) = \int_s^t H(\theta)W(\theta,s)f_jd\theta.
\]
The inequality $\left\Vert W(t,s) f \right\Vert_{B^a}
    \leq \mathrm{e}^{K_a(t - s)}\Vert f \Vert_{B^a}$ holds from \eqref{9.4}. 
So, from Lemma
2.5 in \cite{Ichinose 1995} with $a = b= 1$
  we can see
\[
i\hbar \left(W(t,s)f - f\right) = \int_s^t H(\theta)W(\theta,s)fd\theta
\]
in $B^{a-2}$ and that $W(t,s)f$ is $B^a$-valued continuous and 
$B^{a-2}$-valued continuously
differentiable in $0 \leq s \leq t \leq T.$
  Hence
$\lim_{|\Delta| \rightarrow 0}{\cal C}_{\Delta}(t,s)f \bigl( = W(t,s)f 
\bigr)$ satisfies
\eqref{3.9} with
$u(s) = f$. Thus, we could complete the proof of Theorem 3.1. \par
   We shall consider the proof of Theorem 3.2.  Let $\qts (\theta)$ and
$\overrightarrow{a}^{t,s}_{\Lambda' X,Y}(\theta)$ be the paths defined by 
\eqref{6.1} and
\eqref{6.2}, respectively. For $\xivec_k \in R^2 \ (k \in \Lambda'_1)$ we 
define the path by
\begin{equation} \label{9.11}
\phi_{\xivec_k}^{t,s}(\theta) := \xivec_k + \frac{4\pi 
\rho_k(\qts(\theta))}{|k|^2} \in R^2, \
s
\leq \theta \leq t
     \end{equation}
as in \eqref{3.12}.  The path $\phi_{\xivec_k}^{t,s}(\theta) \in R^2\ (k 
\in \Lambda_1)$ is
defined by \eqref{2.13}.  So from \eqref{2.16}  and (2.17) we have
\[
\xi_{-k}^{(1)} = \xi_{k}^{(1)},\quad  \xi_{-k}^{(2)} = - \xi_{k}^{(2)}.
\]
 For $k \in \Lambda_1$ we can easily
see
\begin{align} \label{9.12}
& |k|^2 \left| \phi_{\xivec_k}^{t,s}(\theta)\right|^2 - 8\pi \rho_k 
(\qts(\theta))\cdot
  \phi_{\xivec_k}^{t,s}(\theta) \notag \\
& = |k|^2 \left|  \phi_{\xivec_k}^{t,s} - \frac{4\pi \rho_k}{|k|^2} \right|^2 -
\frac{16\pi^2 }{|k|^2}|\rho_k|^2 \notag \\
& = |k|^2 \left|\xivec_k\right|^2 - \frac{16\pi^2 }{|k|^2}\left|\rho_k 
(\qts(\theta))\right|^2.
\end{align}
So, the classical action for $\tilde{L}$ defined by \eqref{3.11} is written as
\begin{align} \label{9.13}
& S\left(t,s;\qts,a_{\Lambda 
X,Y}^{t,s},\left\{\phi_{\xivec_k}^{t,s}\right\}_{k \in
\Lambda_1}\right)
\notag \\
& = S_c(t,s;\qts,a_{\Lambda X,Y}^{t,s}) + \frac{(t - s)}{4\pi|V|}
\sum_{k \in \Lambda'_1} |k|^2 |\xivec_k|^2
\end{align}
from \eqref{2.21} and \eqref{3.3}. \par
    Let $\chi_1 \in \Cspace (R^{2N_1})$ with compact support such that 
$\chi_1(0) = 1$.  Let
$\epsilon > 0$ and $ \xi := \left\{\xivec_k\right\}_{k \in
\Lambda'_1} \in R^{2N_1}$.  For $f \in \Sspace(R^{3n+4N})$ we define 
$G_{\epsilon}(t,s)f\ (0
\leq s
\leq t
\leq T)$ by
\begin{equation}  \label{9.14}
        \begin{cases}
            \begin{split}
              & \left(\prod_{j=1}^n\sqrt{\frac{m_j}{2\pi i\hbar(t - s)}}^{\ 
3}\right)
               \sqrt{\frac{1}{2\pi i\hbar|V|(t - s)}}^{\ 4N}
\left(\prod_{k \in \Lambda'_1}  \frac{|k|^2(t -s)}{4i\pi^2\hbar|V|}\right)\\
& \quad\  \times \int \cdots
\int
\mathrm{e}^{i\hbar^{-1}S} \chi_1(\epsilon \xi) f(y,Y)dydY\prod_{k \in 
\Lambda'_1} d\xivec_k ,
                      \end{split}
         & s < t ,
                  \\
\begin{split}
        & f ,\\
             \end{split}
            & s = t,
        \end{cases}
\end{equation}
where $S = S\left(t,s;\qts,a_{\Lambda 
X,Y}^{t,s},\left\{\phi_{\xivec_k}^{t,s}\right\}_{k \in
\Lambda_1}\right)$.
\begin{pro}
Let $f \in B^a(R^{3n+4N}) \ (a = 0,1,2,\dotsc)$.  Then, under the assumptions 
of Theorem 3.1
we have
\begin{equation} \label{9.15}
\lim_{\epsilon \rightarrow 0}G_{\epsilon}(t,s)f = \Cts f
\end{equation}
in $B^a$ for $0 \leq t - s \leq \rho^*$.
\end{pro}
\begin{proof}
In the case of $t = s$ \eqref{9.15} is clear from \eqref{7.6}.  Let $0 < t 
- s \leq
\rho^*$ and $\Sspace(R^{3n+4N})$.  From
\eqref{9.13} we have
\begin{align*}
& G_{\epsilon}(t,s)f = \left(\prod_{j=1}^n\sqrt{\frac{m_j}{2\pi i\hbar(t - 
s)}}^{\ 3}\right)
               \sqrt{\frac{1}{2\pi i\hbar|V|(t - s)}}^{\ 4N} \\
& \times \iint \left(\exp i\hbar^{-1}S_c(t,s;\qts,a_{\Lambda X,Y}^{t,s}) 
\right)
f(y,Y)dydY
\\
& \times \left(\prod_{k \in \Lambda'_1}  \frac{|k|^2(t 
-s)}{4i\pi^2\hbar|V|}\right)
\int \cdots \int \left(\exp \frac{i(t - s)}{4\pi\hbar |V| }
\sum_{k \in \Lambda'_1}|k|^2 |\xivec_k|^2 \right) \chi_1(\epsilon\xi)
\prod_{k \in \Lambda'_1} d\xivec_k.
\end{align*} \par
   Let $\etavec_k := (\eta_k^{(1)},\eta_k^{(2)}) \in R^2$ and $\eta :=
\left\{\etavec_k\right\}_{k
\in
\Lambda'_1}$. We know
\begin{equation} \label{9.16}
  \int_{-\infty}^{\infty} \mathrm{e}^{ia\theta^2}d\theta = \sqrt{ 
\frac{i\pi}{a}}
\end{equation}
for a constant $a > 0$.  So we write
\begin{equation} \label{9.17}
G_{\epsilon}(t,s)f = P_{\epsilon}(t,s)f,
\end{equation}
where
\begin{align} \label{9.18}
p_{\epsilon}(t,s) & = \left(\prod_{k \in \Lambda'_1}  \frac{|k|^2}{i\pi}\right)
\int \cdots \int \left(\exp
i\sum_{k \in \Lambda'_1}|k|^2 |\etavec_k|^2 \right) \notag \\
& \times \chi (\epsilon\sqrt{4\pi\hbar|V|/(t-s)}\eta)
\prod_{k \in \Lambda'_1} d\etavec_k.
\end{align} \par
   We see that
\begin{equation} \label{9.19}
\lim_{\epsilon \rightarrow 0}\, p_{\epsilon}(t,s) = 1
\end{equation}
pointwise. Letting $q_{\epsilon}(t,s) = p_{\epsilon}(t,s) - 1$, we have
\begin{equation*}
  P_{\epsilon}(t,s)f - \Cts f  = Q_{\epsilon}(t,s)f.
\end{equation*}
We consider
\begin{align*}
& \Vert G_{\epsilon}(t,s)f - \Cts f\Vert^2  = \Vert P_{\epsilon}(t,s)f - \Cts
f\Vert^2\\
&  = \Bigl( \bigl(P_{\epsilon}(t,s) - \Cts 
\bigr)^{\dag}\bigl(P_{\epsilon}(t,s) - \Cts\bigr)f
  , f\bigr) \\
& = \bigl(Q_{\epsilon}(t,s)^{\dag}Q_{\epsilon}(t,s)f,f\bigr).
\end{align*}
Hence, we obtain
\eqref{9.15} as in the proof of Theorem 7.3 in the present paper together 
with Lemma 2.2 in
\cite{Ichinose 1995}. See the proof of Lemma 4.1 in \cite{Ichinose 2000} 
for further details.
\end{proof} \par
    We can write \eqref{3.13} as
\begin{equation}  \label{9.20}
     \lim_{\epsilon\rightarrow 0}\, 
G_{\epsilon}(T,\tau_{\nu-1})\chi(\epsilon\cdot)
   G_{\epsilon}(\tau_{\nu-1},\tau_{\nu-2})\chi(\epsilon\cdot)\cdots
G_{\epsilon}(\tau_2,\tau_1)\chi(\epsilon\cdot)
G_{\epsilon}(\tau_1,0)\chi(\epsilon\cdot)f
\end{equation}
in the same way that \eqref{3.8} is written in the above of \eqref{9.2}. 
Integrating by parts
in
\eqref{9.18}, we see that $\sup_{0 < \epsilon \leq 1}|p_{\epsilon}(t,s)|$ 
is finite.  So the
same proof as for \eqref{7.5} shows
\[
\sup_{0 < \epsilon \leq 1}\Vert G_{\epsilon}(t,s)f\Vert_{B^a} \leq C_a 
\Vert f \Vert_{B^a},\ a
= 0,1,2,\dotsc
\]
with constants $C_a$ from \eqref{9.17}.
  Hence, using
\eqref{9.1}, we can prove Theorem 3.2 as in the proof of the convergence of 
\eqref{3.8} to
\eqref{9.2} together with
\eqref{9.15}.
\par
    Finally, we will prove Theorem 3.3.  As in the proof of \eqref{6.15} we get
\begin{align}  \label{9.21}
& \left(\int_{{\pmb q}_{{\pmb z},{\pmb y}}^{t,s}} - \int_{{\pmb q}_{{\pmb 
z},{\pmb
x}}^{t,s}}\right) \left\{\frac{1}{c}A_{\text{ex}}(t,x^{(j)})\cdot dx^{(j)}
  - \phi_{\text{ex}}(t,x^{(j)})dt\right\}
\notag\\
& = \frac{1}{c}(x^{(j)} - y^{(j)}) \cdot \int_0^1 
A_{\text{ex}}\left(s,x^{(j)} - \theta(x^{(j)}
- y^{(j)}) \right) d\theta \notag \\
& - (t - s)(x^{(j)} - y^{(j)}) \cdot \int_0^1 \int_0^1 
\sigma_1E_{\text{ex}}\left(\tau(\sigma),
\zeta^{(j)}(\sigma)\right) d\sigma_1 d\sigma_2
\notag \\
& - \frac{1}{c}\sum_{m=1}^3 (x^{(j)}_m - y^{(j)}_m)
\sum_{l=1}^3 (z^{(j)}_l - x^{(j)}_l)\int_0^1 \int_0^1 \sigma_1 B'_{ml}
\left(\tau(\sigma),
\zeta^{(j)}(\sigma)\right) d\sigma_1 d\sigma_2,
\end{align}
where $\left(B'_{23}(t,x),B'_{31}(t,x),B'_{12}(t,x) \right) =
B_{\text{ex}}(t,x), B'_{lm} = - B'_{ml}$, and $\tau(\sigma)$ and 
$\zeta^{(j)}(\sigma)$
were defined by \eqref{6.11}.  See the proof of Proposition 3.3 in 
\cite{Ichinose 1997} for
further details.  So, we get the equation \eqref{6.18} where the sum over 
$j = 1,2,\dotsc,n$
of \eqref{9.21} multiplied by $m_je_j/(t - s)$ is added to.  Hence, under 
the assumptions of
Theorem 3.3 we obtain the same assertion as in Theorem 3.1 in the same way 
that Theorem 3.1 is
proved.  As in the same way of the proof of Theorem 3.2 we also get the 
same assertion as in
Theorem 3.2 under the assumptions of Theorem 3.3.  Thus, we could complete 
the proof of the
main results.

\begin{thebibliography}{99} %
%
%
%
\bibitem{Arai}
A.  Arai, {\it Fock Space and  Quantum Field} (in Japanese), Nihon Hyoron Co.,
Tokyo,
2000.
%
%
\bibitem{Berezin-Shubin}
 F. A. Berezin and  M. A. Shubin, {\it The Schr\"odinger Equation}, 
 Kluwer Academic
Publishers, Dordrecht, 
  1983.
%
%
\bibitem{Calderon-Vaillancourt}
A. P.  Calder\'on and R. Vaillancourt, On the boundedness of 
pseudo-differential operators,
{\it J. Math. Soc. Japan}
{\bf 23} (1971), 374-378. 
%
%
\bibitem{Dirac}
 P. A. M. Dirac, {\it The Principles of  Quantum Mechanics} 4th ed, 
 Oxford Univ. Press,  London,
1958.
%
%
\bibitem{Fermi}
 E. Fermi, Quantum theory of radiation, {\it Rev. Modern Phys.}  {\bf 4} (1932)
, 87-132. 
%
%
\bibitem{Feynman 48}
  R. P. Feynman, Space-time approach to nonrelativistic quantum mechanics,
   {\it Rev. Modern Phys.}  {\bf 20} (1948), 367-387.
%
%
\bibitem{Feynman 50}
  R. P. Feynman, Mathematical formulation of the  quantum theory of
electrodynamic interaction,  {\it Phys. Rev.}  {\bf 80} (1950), 440-457.
%
%
\bibitem{Feynman-Hibbs}
  R. P. Feynman and A. R. Hibbs,
 {\it   Quantum Mechanics and Path Integrals}, McGraw-Hill, New York,   1965.
%
%
\bibitem{Gelfand-Vilenkin}
 I. M. Gel'fand and  N. Y. Vilenkin,
 {\it  Generalized Functions. Vol. IV, Applications of Harmonic Analysis},
  Academic
Press, 
 New 
York-London, 1964.
%
%
\bibitem{Gustafson-Sigal}
  S. J. Gustafson and I. M. Sigal,
  {\it  Mathematical Concepts of Quantum Mechanics}, Springer, Berlin,   2003.
%
%
\bibitem{Hiroshima}
 F. Hiroshima, Functional integral representation of a model in quantum
electrodynamics,
{\it Rev. Math. Phys.}
{\bf 9} (1997), 489-530.
%
%
\bibitem{Ichinose 1995}
  W. Ichinose, A note on the existence and $\hbar$-dependency  of the
solution of equations in quantum mechanics, {\it Osaka J. Math.} {\bf 32} 
(1995), 
327-345.
%
%
\bibitem{Ichinose 1997}
 W. Ichinose, On the formulation of the Feynman path integral
through broken line paths, {\it Commun. Math. Phys. }{\bf 189} (1997), 17-33.
%
%
\bibitem{Ichinose 1999}
W. Ichinose, On convergence of the Feynman path integral formulated through
broken line
paths, {\it Rev. Math. Phys.} {\bf 11} (1999), 1001-1025.
%
%
\bibitem{Ichinose 2000}
W. Ichinose, The phase space Feynman path integral with gauge invariance and
its convergence, {\it Rev. Math. Phys.} {\bf 12} (2000), 1451-1463.
%
%
\bibitem{Ichinose 2003}
W. Ichinose,  Convergence of the Feynman path integral in the weighted
Sobolev spaces
and the representation of correlation functions, {\it J. Math. Soc. Japan}
{\bf 55} (2003), 957-983.
%
%
\bibitem{Johnson-Lapidus}
 G. W. Johnson and  M. L. Lapidus, {\it The Feynman Integral and 
 Feynman's Operational
Calculus}, Oxford Univ. Press, Oxford,    2000.
%
%
\bibitem{Kumano-go}
 H. Kumano-go, {\it Pseudo-Differential Operators}, MIT Press, Cambridge,
1981.
%
%
\bibitem{Lieb-Loss}
 E. H. Lieb and M. Loss, {\it Analysis}, AMS, Providence,
1997.
%
%
\bibitem{Pauli-Fierz}
  W. Pauli and M. Fierz, Zur theorie der Emission langwellinger 
  Lichtquanten,
{\it Nuovo Cimento}
{\bf 15} (1938), 167-188.
%
%
\bibitem{Reed-Simon}
 M. Reed and B. Simon,    {\it Methods of Modern Mathematical Physics IV:
Analysis of Operators}, Academic Press, New York,   1978.
%
%
\bibitem{Sakurai}
  J. J. Sakurai,  {\it Advanced Quantum Mechanics}, Addison-Wesley,
  Massachusetts,
1967.
%
%
\bibitem{Schwartz}  J. T. Schwartz, {\it Nonlinear Functional Analysis},
Gordon and Breach
Science Publishers, New 
York,   1969.
%
%
\bibitem{Spohn}
H. Spohn,  {\it Dynamics of Charged Particles and Their Radiation Field},
Cambridge University
Press,
Cambridge,
  2004.
%
%
\bibitem{Swanson}
 M. S. Swanson,  {\it Path Integrals and Quantum Processes}, Academic 
Press,
  San Diego, 
1992.
%
%
\end{thebibliography}
\end{document}